\documentclass[11pt, letterpaper]{article}
\usepackage{adjustbox}
\usepackage[ruled, vlined, linesnumbered]{algorithm2e}
\usepackage{amsfonts,amsmath,amsthm,amssymb,dsfont}
\usepackage{array}
\usepackage{bm,bbm}
\usepackage{comment}
\usepackage{enumerate}
\usepackage{float,wrapfig}
\usepackage[T1]{fontenc}
\usepackage{fullpage}
\usepackage{framed}
\usepackage[margin=1in]{geometry}
\usepackage{graphicx} % support the \includegraphics command and options

\usepackage[colorlinks,citecolor=blue,linkcolor=blue,urlcolor=black]{hyperref}
\usepackage[capitalize]{cleveref}

\usepackage{mathtools, mathdots, mathrsfs}
\usepackage{multirow}
\usepackage{natbib}

\usepackage{stmaryrd}
\usepackage{tcolorbox}
\usepackage{tgtermes}
\usepackage{tikz}
\usepackage{thmtools}
\usepackage{thm-restate}
\usepackage[normalem]{ulem}
\usepackage{url}
\usepackage{verbatim}
\usepackage{xspace}
\xspaceaddexceptions{]\}}
\usepackage{xfrac,nicefrac}

\usetikzlibrary{hobby}
\usetikzlibrary{decorations.markings}
\usetikzlibrary{quotes,angles}

\DeclarePairedDelimiter\abs{\lvert}{\rvert}%
\DeclarePairedDelimiter\norm{\lVert}{\rVert}%
\DeclarePairedDelimiter\greaterInt{\lceil}{\rceil}%
\DeclarePairedDelimiter\pair{<}{>}%
\DeclarePairedDelimiter\targetPairRaw{(}{)}%

\makeatletter
\newcommand{\target}{\@ifnextchar\bgroup{\target@with}{\target@without}}
\def\target@with#1{\Tilde{U}\targetPairRaw{#1}}
\def\target@without{\Tilde{U}}

\let\oldabs\abs
\def\abs{\@ifstar{\oldabs}{\oldabs*}}
\let\oldnorm\norm
\def\norm{\@ifstar{\oldnorm}{\oldnorm*}}
\let\oldgreaterInt\greaterInt
\def\greaterInt{\@ifstar{\oldgreaterInt}{\oldgreaterInt*}}
\let\oldpair\pair
\def\pair{\@ifstar{\oldpair}{\oldpair*}}

\newcommand{\frontierPair@star}[1]{\langle#1\rangle}
\newcommand{\frontierPair@nostar}[1]{\left<#1\right>}
\def\frontierPair{\@ifstar{\frontierPair@star}{\frontierPair@nostar}}
\makeatother

\theoremstyle{plain}

\newtheorem{theorem}{Theorem}[section]
\newtheorem{lemma}[theorem]{Lemma}
\newtheorem{observation}[theorem]{Observation}

\theoremstyle{definition}

\newtheorem{remark}[theorem]{Remark}

\def\ShowAuthNotes{1}
\ifnum\ShowAuthNotes=1
\newcommand{\authnote}[2]{\ \\ \textcolor{red}{\parbox{0.9\linewidth}{[{\footnotesize {\bf #1:} { {#2}}}]}}\newline}
\else
\newcommand{\authnote}[2]{}
\fi

\newcommand{\bydef}{\stackrel{\mathrm{def}}{=}}
\renewcommand{\epsilon}{\varepsilon}

\renewcommand{\hat}{\widehat}

\newcommand{\expectB}{B}
\newcommand{\anyB}{B}

\newcommand{\actualB}{B'}
\newcommand{\actualU}{U}

\newcommand{\ds}{\mathcal{D}}

\newcommand{\pull}{\textsc{Pull}}

\newcommand{\mergeds}{\textsc{Merge}}

\newcommand{\ins}{\textsc{Insert}}
\newcommand{\decrkey}{\textsc{Decrease-Key}}
\newcommand{\extract}{\textsc{Extract-Min}}

\newcommand{\nonEmpty}{\textsc{Non-Empty}}

\newcommand{\True}{\mathtt{True}}
\newcommand{\False}{\mathtt{False}}

\newcommand{\dirtyPj}{J}
\newcommand{\distV}{\mathtt{dis}}

\SetKwInput{AlgInput}{Input}
\SetKwInput{AlgOutput}{Output}
\SetKwInput{Require}{Require}
\SetKwInput{Ensure}{Ensure}
\SetKwInput{Note}{Note}
\SetKwProg{Fn}{Function}{:}{}
\SetKwProg{Alg}{Algorithm}{:}{}
\SetKw{Tuple}{Tuple}
\SetKwFunction{Relax}{Relax}
\SetKwFunction{Addition}{Addition}
\SetKwFunction{Comparison}{Comparison}
\SetKwFunction{FnDeactivate}{Deactivate}
\SetKwFunction{FindPivots}{FindPivots}
\SetKwFunction{BMSSP}{BMSSP}
\SetKwFunction{UpdatePj}{Update-Pj}
\SetKwFunction{Partition}{Partition}
\SetKw{LogicAnd}{and}
\SetKw{LogicOr}{or}
\SetKw{KeyWhile}{while}
\SetKw{KeyContinue}{continue}
\SetKw{KeyBreak}{break}
\SetKw{KeyWhere}{where}
\SetKw{KeyIf}{if}
\SetKw{KeyFrom}{from}
\SetKw{KeyTo}{to}
\DontPrintSemicolon

\title{A Faster Directed Single-Source Shortest Path Algorithm}

\author{
    Ran Duan \thanks{Tsinghua University. Email: duanran@mail.tsinghua.edu.cn,
    ylh21@mails.tsinghua.edu.cn.}
    \and Xiao Mao \thanks{Stanford University. Email: matthew99a@gmail.com.}
    \and Xinkai Shu \thanks{Max Planck Institute for Informatics. Email: xshu@mpi-inf.mpg.de.}
    \and Longhui Yin \footnotemark[1]
}

\date{February 2026}

\begin{document}
\maketitle

\begin{abstract}
    This paper presents a new deterministic algorithm for single-source shortest paths (SSSP) on real non-negative edge-weighted directed graphs, with running time $O(m\sqrt{\log n}+\sqrt{mn\log n\log \log n})$, which is $O(m\sqrt{\log n\log \log n})$ for sparse graphs. This improves the recent breakthrough result of $O(m\log^{2/3} n)$ time for directed SSSP algorithm [Duan, Mao, Mao, Shu, Yin 2025]. %and almost matches the time bound for undirected SSSP algorithm [Duan, Mao, Shu, Yin 2023][Yan 2025].

    % O(m\sqrt{n}(1+\sqrt{\frac{n\log \log n}{m}}))

\end{abstract}

\section{Introduction}
The single-source shortest path (SSSP) problem is one of the most foundational problems in graph theory,
whose algorithms have significant improvements since the 1950s.
Given a graph $G=(V,E)$ with $n$ vertices and $m$ non-negative real-weighted edges, the goal is to compute the distance from a source vertex $s$ to every vertex $v\in V$.
We work under the \emph{comparison-addition} model,
where only comparison and addition on edge weights are allowed --
a natural assumption that align with real-weighted inputs.

The textbook algorithm by Dijkstra \cite{Dij59}, when implemented with a Fibonacci heap \cite{FT87}, solves the problem in $O(m+n\log n)$ time.
In addition to computing the distances, Dijkstra's algorithm also produces the ordering of vertices. Haeupler, Hlad\'{\i}k, Rozho\v{n}, Tarjan and T\v{e}tek \cite{HHRTT24} demonstrated that Dijkstra's algorithm is universally optimal if this order is required. On the other hand, if only the distances are needed, there were no substantial improvements until recently.
Duan, Mao, Shu and Yin \cite{DMSY23} proposed a randomized algorithm with $O(m\sqrt{\log n\log \log n})$ running time for undirected graphs, which was subsequently derandomized by Yan \cite{Yan25}.
Later, a groundbreaking result by Duan, Mao, Mao, Shu and Yin \cite{SSSP25} introduced a deterministic $O(m\log^{2/3} n)$-time algorithm for directed graphs, first to break the $O(n\log n)$ time bound in sparse graphs.

In this work, we further improve the running time of SSSP to $O(m\sqrt{\log n}+\sqrt{mn\log n\log \log n})$ on directed graphs, matching the previous results for undirected graphs.
\begin{theorem}\label{thm:main}
Single-source shortest path on directed graphs with real non-negative edge weights can be solved in deterministic $O(m\sqrt{\log n}+\sqrt{mn\log n\log \log n})$ time.
\end{theorem}

An equivalent statement is, $O(m\sqrt{\log n})$ time for $m\geq n\log\log n$, $O(\sqrt{mn\log n\log \log n})$ for $m<n\log\log n$, and $O(n\sqrt{\log n\log \log n})$ for $m=O(n)$.

\subsection{Technical Overview}
\label{sec:overview}
Let us begin with Dijkstra's algorithm,
which maintains a ``frontier'' $S$ of vertices, such that if a vertex is ``incomplete'' --- meaning its current distance label is greater than its true distance --- then its shortest path visits some complete vertex in $S$.
In other words, the frontier encloses a set of vertices sufficient for finding the shortest paths of all remaining incomplete vertices.

Dijkstra's algorithm each time extracts from the frontier the vertex with smallest distance label, which is guaranteed to be complete. It then relaxes all out-going edges and adds these neighbors to frontier.
This produces an ordering of all vertices by distances, creating an inherent $\Omega(n\log n)$ sorting barrier.

%Our algorithm builds upon the general framework in \cite{SSSP25}, with significant improvements inspired by \cite{SSBP18}.
The breakthrough in \cite{SSSP25} overcomes this barrier by introducing a divide-and-conquer procedure with $(\log n)/t$ levels, namely Bounded Multi-Source Shortest Path (BMSSP), which operates on a ``frontier'' $S$ with respect to a bound $B$. BMSSP aims to settle all vertices in $\target = \target{B, S}= \{v: \distV(v) < B \text{ and its shortest path visits some vertex in }S\}$.
It recursively invokes subcalls of BMSSP on sets consisting of the smallest $1/2^t$ fraction of vertices extracted from the frontier.
A na\"ive implementation would still spend $\Theta(t)$ time on each frontier vertex, yielding a total running time of $\Theta(\log n)$ per vertex.

To solve this problem, \cite{SSSP25} reduces the size of frontier $S$ by a factor $1/k$ regarding to the target set $\target$,
by running Bellman-Ford algorithm for $k$ iterations. Any remaining incomplete vertex must lie in the shortest path tree of some $x\in S$ with size at least $k$. Such vertices $x$ are called ``pivots'', and their number is bounded by $\abs*{\target} / k$.
However, this approach still poses a bottleneck since Bellman-Ford requires $\Omega(k)$ time per vertex on average.
This cost seems unavoidable, because the algorithm needs to determine, for each vertex $v\in\target$ which pivot's shortest path tree it belongs to.
% A naive alternative proposal could be to start a Dijkstra search from each vertex $x\in S$, but this does not work because if the search from $x$ and $y$ meed at some vertex $z$, $z$ and vertices under $z$ may need to be scanned twice or even more.
% Few algorithms does this better than Bellman-Ford algorithm.

A key observation is that, such explicit mapping is unnecessary.
Our new insight is that we can construct a spanning forest that connects the frontier $S$ via local Dijkstra searches from $S$, and partition it into edge-disjoint subtrees of sizes $\Theta(k)$. This idea is inspired by \cite{SSBP18} that tackles the single-source bottleneck path (SSBP) problem, and reduces the cost of finding pivots from $O(k)$ to $O(\log k)$ per vertex. Then, for each subtree, the algorithm only tracks the vertex with the smallest distance label, referred to as a ``pivot'', which serves as a handle for the vertices in its corresponding subtree.
Then we are able to maintain $1/k$ of the vertices in $S$ in the partial-sorting data structure, and meanwhile efficiently select a minimal subset of the whole frontier. For other vertices, we only need to check whether they fit the range of the subset without sorting them. Then the time for inserting a vertex in the data structure can be amortized to an edge on the subtree containing it, but that edge is not supposed to be in recursive subcalls. 
So the overall running time is roughly $O(m\sqrt{\log n})$. However, in this recursive procedure for shortest path, the upper bound $B$ used for each invocation is dynamically generated. If the range by bound $B$ contains too many vertices with respect to $S$, the algorithm can only perform partial execution, and some vertices may be reinserted into $S$ after extraction. This behavior must be handled carefully to achieve the running time mentioned above.
% This behavior must be handled carefully to achieve the running time mentioned above.

%It is also the core insight that allows us to improve the algorithm in \cite{SSSP25}.

% Compared to \cite{SSBP18}, our case involves significantly more subtlety.
% The upper bound $B$ used for each invocation of BMSSP is dynamically generated. Therefore, we cannot directly control the size of the corresponding $\target$. If $\target$ contains too many vertices with respect to $S$, the algorithm can only perform partial execution, and some vertices may be reinserted into $S$ after extraction.
% This behavior must be handled carefully to achieve the running time mentioned above.

% Compared to \cite{SSSP25}, o
% Our another important observation is that under a full execution, a non-pivot vertex\footnote{in \cite{SSSP25}, a vertex from which Bellman-Ford steps does not reach to $k$ vertices; in our paper, a vertex from which local Dijkstra does not reach to $k$ vertices.}, once complete in the current round, does not need to be processed again even in subsequent rounds until a new valid relaxation --- which, on the contrast, \cite{SSSP25} always had them reinserted back.
% Consequently, we are able to associate such vertices' processing time globally with its degree, such that our algorithm is able to improve time complexity.
% Besides, we also simplified the data structure by leveraging information in sub-calls, supplemented and extracted more properties on relaxations, tie-breaking and frontier to make the algorithm easier to be understood.

\subsection{Related Works}
Single-source shortest path has been extensively studied in various special cases. Pettie and Ramachandran \cite{PR05} proposed an algorithm running in $O(m\alpha(m,n)+\min\{n\log n,n\log\log r\})$ time for undirected graphs, where $\alpha$ is the inverse-Ackermann function and $r$ is the max-min ratio of edge weights.
In word-RAM model with bounded integer edge weights, a sequence of works \cite{FW93,FW94, Thorup96, Raman96, Raman97, TM00, HT20} culminated in Thorup's linear-time algorithm for undirected graphs \cite{Thorup00} and $O(m+n\log\log\min\{n,C\})$-time for directed graphs \cite{Thorup04}, where $C$ is the maximum edge weight.

If negative edge weights are allowed, the Bellman-Ford algorithm \cite{Bellman1958} needs $O(mn)$ time. Beyond this, there are nearly linear time algorithms for integer weights (with respect to bit complexity) \cite{flow,BNW22, BCF23}, and strongly subcubic time algorithms for real weights \cite{Fineman24,HJQ25b,HJQ25}.
\section{Preliminaries}

The single-source shortest path problem is defined on a directed graph $G = (V, E)$ with a non-negative weight function $w: E \to \mathbb{R}_{\geq 0}$ (typically denoted $w_{uv}$) assigned to each edge. Let $n = \abs{V}$ and $m = \abs{E}$ denote the number of vertices and edges respectively. Given a source $s\in V$, the goal is to compute the length of the shortest path from $s$ to every vertex $v\in V$.
Without loss of generality we assume that
every vertex is reachable from $s$,
so $m \geq n - 1$.

For clarity, we restate two relevant definitions.
A \emph{directed tree} is a directed acyclic graph whose underlying undirected graph is a tree.
An \emph{arborescence} rooted at some vertex $v$ is a directed tree such that every other vertex can be reached from $v$ via exactly one directed path.

\subsection{Degree reduction}

% It is beneficial for our algorithm to constrain the degree of the graph (as also done in \cite{DMSY23,SSSP25}).
For any degree bound $3 \leq \delta \leq \frac{m}{n}$, we preprocess $G$ in $O(m)$ time to a graph with $O(m)$ edges, $O(m/\delta)$ vertices, and maximum (in- and out-) degree bounded by $\delta$, while preserving all shortest path lengths.
A classical way to achieve this, similar to that in \cite{Frederickson83}, proceeds as follows:
\begin{itemize}
    \item Substitute each $v$ with a zero-weighted cycle $C_v$ of $\greaterInt{\frac{\Delta_v}{\delta-2}}$ vertices, where $\Delta_v$ denotes the degree of $v$. For each $v$'s neighbor $u$ (either incoming or outgoing), assign one vertex in $C_v$ to represent the edge between them, denoted as $x_{vu}$; make sure each vertex in $C_v$ represents at most $\delta-2$ different $x_{vu}$.
    \item For every edge $(u, v)\in G$, add a directed edge from vertex $x_{uv}$ to $x_{vu}$ with weight $w_{uv}$.
\end{itemize}

Note that when $m\geq n\log n$, Dijkstra's algorithm implemented with Fibonacci heap already reaches the optimal time $O(m+n\log n) = O(m)$, thus it suffices to consider the regime $m\leq n\log n$ in our algorithm.

\subsection{Comparison-Addition Model}

Our algorithm works under the \emph{comparison-addition} model, in which edge weights can only be compared or added, each taking $O(1)$ time. No other arithmetic or algebraic operations on edge weights are allowed.

\subsection{Edge Relaxation and Tie-Breaking}

We maintain global distance labels $d[v]$ for all vertices $v\in V$.
Initially $d[s]=0$ and $d[v]=\infty$ for all $v\neq s$.
All updates to $d[v]$ are through relaxation: $d[v] \gets \min\{d[v], d[u] + w_{uv}\}$ of some edge $(u, v)\in E$.
A relaxation is \emph{valid}, if $d[v]$ was updated, namely $d[u] + w_{uv} \leq d[v]$.
In our algorithm, we also introduce an upper bound $B$ and require $d[u] + w_{uv} < B$ to prevent immature modification.
Hence, $d[v]$ is non-increasing and always corresponds to the length of some path from $s$ to $v$.
Let $\distV(v)$ denote the length of the true shortest path from source $s$ to $v$, then $d[v] \geq \distV(v)$ always holds.

Similar to previous works \cite{DI04APSP,SSSP25}, for clarity in presenting the algorithm, we define a tie-breaking rule for comparing distance labels in order to:
\begin{enumerate}
    \item Keep the structure of the shortest path tree throughout the algorithm;
    \item Establish a relative ordering among the vertices with identical $d[]$.
\end{enumerate}

In principle, we treat a path of length $l$ that traverses $q$ edges, represented by the vertices $v_0=s, v_1, ..., v_q$, as a tuple $\langle l, q, v_q, v_{q-1},...,v_0\rangle$, where vertices are listed in reverse order, and sort them based on the lexicographical order.
However, it turned out that the first 4 items of the above tuple are sufficient for tie-breaking: path length ($\mathtt{length}$), number of edges ($\mathtt{nEdges}$), $v$ itself ($\mathtt{curr}$), and $v$'s predecessor ($\mathtt{pred}$).
% When comparing two distance labels, we perform lexicographical comparison on the two tuples in $O(1)$ time.
% Minimizing $\mathtt{nEdges}$ ensures that zero-weighted loops are avoided, while $\mathtt{prev}$, $\mathtt{curr}$ serves as additional tie-breakers when necessary.
This is because our algorithm performs such comparisons in the following scenarios:
\begin{description}
    \item[Relaxing an edge $(u,v)$:] If $u\neq d[v].\mathtt{pred}$, even if $\mathtt{length}$ and $\mathtt{nEdges}$ are tied, it suffices to compare $u$ and $d[v].\mathtt{pred}$; if $u=d[v].\mathtt{pred}$, then $d[u]$ has been updated to currently ``shortest'' and $d[v]$ must get updated accordingly;
    \item [Comparing two different {$d[u]$ and $d[v]$} for $u \neq v$:] In this case $\mathtt{curr}$ acts as a tie-breaker.
    \item [Upper bound $\expectB$:] In the algorithm the upper bound $\expectB$ is also represented by such a 4-tuple. %Such bounds and relaxations besides \FindPivots are always formed by a complete shortest path plus an additional edge, so it is easy to compare them using 4-tuples.
\end{description}

Algorithm \ref{alg:relaxation} illustrates how an edge $(u, v)$ is relaxed.
Note that an upper bound $\expectB$ is introduced to prevent premature update of vertices whose tentative distances exceed the current bound in the main algorithm.

\begin{algorithm}[ht]
\label{alg:relaxation}
\caption{Relaxation of an edge $(u,v)\in E$ with an upper bound $\expectB$}
% \Fn(\tcp*[f]{$(u, v)$ must be an edge in $E$})
\Fn
{\Relax{vertex $u$, vertex $v$, upper bound $\expectB$}}{
    \tcp{$\expectB$ is introduced to prevent early update of vertices out of bound}
    \If { $d[u] + w_{uv} \leq d[v]$ \LogicAnd $d[u] + w_{uv} < \expectB$ }{
        $d[v] \gets d[u] + w_{uv}$\;
        \Return $\True$\;
    }
    \Return $\False$\;
    % $\mathtt{IsValid} \gets \mathbb{I}[d[u] + w_{uv} \leq d[v]]$\;
    % $d[v] \gets \min\{d[v], d[u] + w_{uv}\}$\;
    % \Return $\mathtt{IsValid}$\;
}
\Tuple $d$ \textbf{:} $\pair{\mathtt{length}, \mathtt{nEdges}, \mathtt{curr}, \mathtt{pred}}$\;
\Fn{\Addition{distance label $d[u]$, edge weight $w_{uv}$}}{
    \Return $\pair{d[u].\mathtt{length} + w_{uv}, d[u].\mathtt{nEdges} + 1, v, u}$
}
\Fn{\Comparison{distance label $d[u]$, distance label $d[v]$}}{
    \Return Lexicographically compare $\pair{\mathtt{length}, \mathtt{nEdges}, \mathtt{curr}, \mathtt{pred}}$ of $d[u]$ and $d[v]$\;
}
\end{algorithm}

% We show that this scheme is sufficient for path length tie-breaking in SSSP by an inductive argument:
% assume that for any predecessors $u_1, u_2, \cdots, u_q$ of $v$, the shortest paths ending at each $u_i$ are already uniquely determined. For any paths of equal $\mathtt{length}$ terminating at $v$, each such path must pass through some predecessor $u_i$. Since the shortest path to each $u_i$ is unique by the induction hypothesis, there can be at most one shortest path to $v$ passing through each $u_i$.
% These candidate paths can then be uniquely ordered by $\mathtt{nEdges}$ and $\mathtt{prev}$. Finally, $\mathtt{curr}$ is used only to break ties between two vertices that share the same predecessor and have equal-length paths. {\color{red} How to compare without RECURSION?}

%Finally, $\mathtt{curr}$ is used only when comparing two vertices with the same predecessor and equal-length paths.
%Assuming no ambiguity over the paths which ends at predecessors $u_1, u_2, \cdots, u_q$ of $v$, for the paths of equal $\mathtt{length}$ which ends at $v$, they pass through some $u_i$, so there is only one shortest path passing through each $u_i$, and they can be uniquely ordered by $\mathtt{nEdges}$ and $\mathtt{prev}$. Finally, $\mathtt{curr}$ is used only when comparing two vertices with the same predecessor and equal-length paths.

\subsection{Completeness and Frontier}
For a vertex $v$, if $d[v] = \distV(v)$, we say $v$ is \emph{complete}.
If all vertices in a set are complete, we say that set is complete.
% Throughout the algorithm, $d[v]$ is updated iteratively while $\distV(v)$ remains unknown but fixed, determined by input.
% In the algorithm's operations, we only refer to $d[v]$;
% however, in the analysis we may refer to the true distance $\distV(v)$ and $v$'s completeness.
We need to specify the time step when referring to $d[v]$ and completeness because they are sensitive to the algorithm progress.
% As $d[v]$ upper bounds $\distV(v)$, if at some moment $v$ is complete, it continues to be complete ever since.
% and all vertices along the shortest path from $s$ to $v$ must also be complete.

%\subsection{Frontier and Dijkstra-like Algorithm Framework}
%\label{para:frontier-and-dijkstra-like-alg-framework}

Let $S\subseteq V$ be a set of vertices and $\expectB$ be a bound. We define $\target = \target{\expectB, S}$ the set of every vertex $v$ such that $\distV(v) < \expectB$ and the shortest path of $v$ visits some vertex in $S$. The goal is to make $\target$ complete.

Additionally, for a vertex set $U \subseteq V$ and any bound $\anyB$, define $d_{\anyB}[U] = \min(\{\anyB\}\cup \{d[x]: x\in U\})$ and $\distV_{\anyB}(U) = \min(\{\anyB\}\cup \{\distV(x): x\in U\})$. ($\anyB$ is introduced here in case $U$ is an empty set.)

We introduce \emph{frontier}, which is crucial for understanding our algorithm.
This concept also appeared in \cite{SSSP25}: % and most Dijkstra-like algorithms in equivalent forms:
at any moment and, for two sets $X, Y \subseteq \target$, we say $\frontierPair{X, Y}$ is a frontier for $\target$, if any $v\in \target$ satisfies at least one of the following:
\begin{enumerate}
    \item $v\in X$ and $v$ is complete;
    \item the shortest path of $v$ visits some \emph{complete} vertex (possibly itself) in $Y$.
\end{enumerate}

% \paragraph{Remark.}
% In this definition
Note $X$ is not necessarily complete; but any incomplete vertex of $X$ must satisfy the second condition.
% $\frontierPair{\emptyset, S}$ is not trivially a frontier for $\target{\expectB, S}$.
% Specifically, it requires that for any vertex $v\in \target$ ($\distV(v) < \expectB$) if the shortest path of $v$ visits some vertex in $S$, such a path must also visit a \emph{complete} vertex in $S$.
Based on the definition, we can make the following observations:
\begin{observation}
\label{lemma:frontier-observation}
At any given moment, if $\frontierPair{X, Y}$ is a frontier for $\target=\target{\expectB, S}$:
\begin{enumerate}
    \item $\frontierPair{X, Y}$ continues to be a frontier for $\target$ ever since;
    \item $d_{\expectB}[Y] = \distV_{\expectB}(Y)$;
    \item $\target{\expectB, Y} \subseteq \target{\expectB, S}$;
    \item $\frontierPair{\emptyset, Y}$ is a frontier for $\target{\expectB, Y}$;
    % noticing that $\target{\expectB, Y} \subseteq \target{\expectB, S}$;
    \item $\target{d_{B}[Y], S} \subseteq X$ and $\target{d_{B}[Y], S}$ is complete;
    %\item Denote $Y^*$ the set of complete vertices in $Y$. Then $\frontierPair{X, Y^*}$ is also a frontier of $\target$.
\end{enumerate}
\end{observation}

Proof sketch:
the first observation lies from the fact that when a vertex becomes complete, it continues to be complete ever since.
The second observation is because the vertex with minimum distance label $d[]$ in $Y$ must be complete.
The third observation is because: for any vertex in $\target{\expectB, Y}$, its shortest path visits some vertex $y$ in $Y$; and the shortest path of $y$, as a prefix of the shortest path of $v$, owing to $Y\subseteq \target = \target{\expectB, S}$, visits some vertex in $S$, so the shortest path of $v$ visits $S$;
The fourth observation is because: for any incomplete vertex in $\target{\expectB, Y}$, it is also in $\target{\expectB, S}$, then its shortest path visits some complete vertex in $Y$; for any complete vertex $v$ in $\target{\expectB, Y}$, by definition its shortest path visits some vertex $y\in Y$, and $y$ must be complete as $v$ is complete.
The fifth is because if the distance of a vertex is less than $d_{\expectB}[Y]$, then it cannot visit any complete vertex in $Y$.

\subsection{A Framework of Dijkstra-like Algorithms}\label{para:frontier-and-dijkstra-like-alg-framework}
%In this section 
We found that Dijkstra's algorithm, \cite{SSSP25} and our algorithm can be restated under a unified framework: %, and almost all Dijkstra-like algorithms:
\begin{enumerate}[(1)]
    \item Given a frontier $\frontierPair{X, Y}$ for $\target=\target{\expectB, S}$. Our goal is to make $\target$ complete;
    \item Identify a complete set $Z\subseteq \target$; or make $Z$ complete, possibly after performing some relaxations;
    \item Relax the outgoing edges from $Z$ and denote $R$ as the set of out-going neighbors of $Z$ but not in $Z$ with distance less than $\expectB$, and relaxations from $Z$ to them are valid.
    Clearly $R\subseteq \target$;
    \item Update to a new frontier for $\target$ as $\frontierPair{X', Y'} \bydef \frontierPair{X \cup Z, (Y\setminus Z)\cup R}$; 
    \item Return to the first step and repeat the process until $Y = \emptyset$, then $\target = X$ is complete.
\end{enumerate}

% We show that for any $Z\subseteq\target$, once $Z$ becomes complete and its outgoing edges are relaxed, $\frontierPair{X', Y'}$ is indeed a frontier for $\target$.
We show that for $Z$ given by (2), after operations in (3), $\frontierPair{X', Y'}$ is indeed a frontier for $\target$.
Since adding vertices never violates frontier conditions,
% but removing vertices could,
it suffices to verify for $v\in \target$ whose shortest path visits some complete $u\in Z$ such that $u$ is in $Y$ but removed from $Y'$, $v$ still satisfy one of the two frontier conditions.
\begin{itemize}
    \item If $v\in Z$, we are already done;
    \item If $v\notin Z$, since $u\in Z$, along the shortest path from $u$ to $v$, there must exist two adjacent vertices $x$ and $y$ such that $x\in Z$ and $y\notin Z$; then $y$ is in $R$ and added to $Y'$, ensuring that the shortest path of $v$ still visits a complete vertex $y\in Y'$.
\end{itemize}

The efficiency of the algorithm depends critically on how we can efficiently identify $Z$ or make $Z$ complete.
In Dijkstra's algorithm, the smallest vertex $\arg\min_{x\in Y}d[x]$ is guaranteed to be complete. Therefore, the algorithm organizes $Y$ using a priority queue and simply selects such $x$ to form $Z$.
% Many previous works focused on constructing such a set $Z$ more efficiently.
In \cite{SSSP25} and our approach, the algorithm repeatedly selects a subset of smallest vertices $S_i$ from the frontier, and recursively invoke a lower-level algorithm on $S_i$ to obtain a complete set $\actualU_{i}$, which is then used to update the frontier.

\section{The Main Result}

Recall that we work on a preprocessed graph of $O(m)$ edges, $O(m/\delta)$ vertices with max degree $\delta$.
Following the approach in \cite{SSSP25}, our algorithm is based on a divide-and-conquer framework in which, at each step, we attempt to extract the smallest $1/2^t$ fraction of vertices from the frontier ($t \leq \log n$).
% We invoke recursive sub-calls on it to make some set complete and update the frontier.

We fix parameter $t$ and define $k = \greaterInt{t / \log t}$, also ensuring $\delta \leq \log k = O(\log \log n)$ holds.\footnote{The algorithm also works when $\delta > \log k$. Setting $\delta = \Theta(\min\{\frac{m}{n}, \log\log n\})$ optimizes the time bound, so we require $\delta \leq \log k$ to simplify the analysis of time complexity.}
A more precise choice of $t$ will be given later; roughly $t \approx \sqrt{\log n}$.
We are now ready to present our main lemma.

\begin{lemma}[Bounded Multi-Source Shortest Path]
\label{lemma:BMSSP}
Given an upper bound $\expectB$, a vertex set $S$ of size at most $t^22^{lt}$, and an integer $l\in [0, \greaterInt{(\log n)/ t}]$.
Suppose $\frontierPair{\emptyset, S}$ is a frontier for $\target = \target{\expectB, S}$.

\BMSSP{$\expectB, S, l$} (\cref{alg:BMSSP}) outputs a new bound $\actualB(\leq \expectB)$,
a vertex set $\actualU\subseteq \target$ of size $O(t^32^{lt})$,
a set of vertices $\ds$ arranged by data structure as in \cref{lemma:data-structure} parameterized by $M := t2^{(l-1)t}$,
% a vertex set $R$ with multiplicities,
in $O(\abs{\actualU}(l\log t + \delta t))$ time.
After running \cref{alg:BMSSP}, $\actualU = \target{\actualB, S}$ is complete and $\frontierPair{\actualU, \ds}$ is a frontier for $\target$. Moreover, one of the following is true:
\begin{itemize}
    \item Full execution ($\actualB = \expectB$): $\ds = \emptyset$ and $\actualU = \target$;
    \item Partial execution ($\actualB < \expectB$):  $\abs{\actualU} = \Theta(t^32^{lt})$.
\end{itemize}
\end{lemma}

In the top layer, we call \cref{alg:BMSSP} with parameters $S = \{s\}$, $\expectB = \infty$ and $l = \greaterInt{(\log n)/t}$.
Since $\abs{\actualU} \leq \abs{V} = o(t^3n)$, it is a full execution with all vertices complete in the end.
Furthermore,
% in the top layer,
the degree of each vertex in $V$ is at most $O(\delta)$.
% (\cref{lemma:bmssp-time-analysis}).
Thus the total running time is $O(m(t + \frac{\log n \log t}{\delta t}))$.

With $t = \sqrt{\log n \log \log n / \delta}$ and $\delta = \frac{1}{4} \min\{\frac{m}{n}, \log \log n\}$, it reduces to $O(m\sqrt{\log n \log \log n / \delta})$.
If $\frac{m}{n} \geq\log \log n$, it is $O(m\sqrt{\log n})$.
If $\frac{m}{n} < \log \log n$, it is $O(\sqrt{mn\log n\log\log n})$.

\subsection{Finding Pivots}

Recall that we hope to manage the frontier $\frontierPair{\emptyset, S}$ more efficiently. 
As mentioned in \cref{sec:overview}, Algorithm \ref{alg:find-pivots} tries to connect $S$ with edge-disjoint $\Theta(k)$ sized sub-trees. We perform a local Dijkstra search from every vertex $x\in S$, only adding edges with valid relaxation into the sub-tree. There are three possible cases:
\begin{itemize}
    \item If at least $k$ vertices are found, we stop and record the result as a directed tree $\bar{F}_j$;
    \item If the search encounters an vertex that already belongs to an existing $\bar{F}_j$ (i.e., it overlaps with an existing sub-tree), we stop and merge the current explored sub-tree with that $\bar{F}_j$;
    \item Otherwise the search fails to reach $k$ vertices from $x$, we record it as an arborescence $W_j$.
\end{itemize}
Finally we partition each $\bar{F}_j$ into $\Theta(k)$ sized edge-disjoint sub-trees in linear time as described in \cref{sec:partition}.

An important observation is that, under a full execution, if the search fails to reach $k$ vertices from $x$, we no longer need to initiate any further search from $x$ at current value $d[x]$ with bound $\expectB$.
\begin{itemize}
    \item If $x$ is currently complete, every vertex $v$ with $\distV(v) < \expectB$ whose shortest path visits $x$ is already reached and updated;
    \item Otherwise, the shortest path of $x$ passes through some other complete vertex in $S$. We can rely on that vertex instead of $x$ for subsequent shortest path relaxations.
\end{itemize}

%In fact, another searching based on the current value $d[x]$ cannot cause any new update which we have already experienced.
Therefore, $x$ is no longer useful for relaxation;
unless $d[x]$ is updated by another relaxation and then $x$ reappears in $S$ in a subsequent call.
Intuitively, each search from $x$ failing to reach $k$ vertices is attributed to one of $x$'s incoming edges, so it occurs at most $\delta$ times. This statement is formally proved in \cref{lemma:bmssp-time-analysis}.
% As we will show in \cref{lemma:bmssp-time-analysis}, this happens at most $\delta$ times for any $x$.

%{\color{cyan} This observation effects on Line 21 of Algorithm 3 in \cite{SSSP25} and its corresponding position \cref{line:bmssp:add-back-Si} in our \cref{alg:BMSSP}: to add back remaining $S_i$ which are in the interval $[\actualB_i, \expectB_{i})$, may cause \cref{alg:find-pivots} called on a vertex for multiple times (a vertex pulled into $S_i$, added back to $\ds$ and then pulled again, added back again, repeat, $\cdots$); and in our algorithm we do not add back such vertices.}

\begin{algorithm}[ht]
\caption{Finding Pivots}
\label{alg:find-pivots}
\AlgInput{Upper bound $\expectB$, vertex set $S$}
\AlgOutput{Vertex sets $\{P_j\}_{j=1}^{p}$, $Q$, $W$}
% \Require {$\frontierPair{\emptyset, S}$ is a frontier for $\target{\expectB, S}$.}
% \Ensure {$S$ is a disjoint union of $\{P_j\}_{j=1}^{r}$ and $Q$.}
% \Ensure {$\frontierPair{W, P}$ is a frontier for $\target{\expectB, S}$. }
% \Ensure {Each $P_j$ is connected by a directed tree $F_j$; $F_j$'s are edge-disjoint; $\abs{F_j} = \Theta(k)$}
% \Ensure {$P_j \subseteq F_j \cap S$; $F_j\subseteq \target$ ; $\abs{F_j} = \Theta(k)$; $\abs{P_j} = O(k)$; $\abs{W_j} < k$.}
% \Ensure {Running time: $O(1)$ per visited edge; $O(\log k) = O(\log t)$ per visited vertex; $O(k\log k)$ per root of $W_j$, amortized to its root.}
\Alg{\FindPivots{upper bound $\expectB$, vertex set $S$}}{
    \ForAll {$x \in S$}{
        \lIf{$x$ is in any existing $\bar{F}_j$}{Skip $x$ and proceed to the next vertex in $S$}
        % \lIf{$x$ is not \activeV\ }{Skip $x$ and proceed to the next vertex in $S$}
        Initialize Fibonacci heap $H\gets \{\pair{x, d[x]}\}$ and subgraph $K\gets \{x\}$\;
        \While(\tcp*[f]{Local Dijkstra search from $x$}){$H.\nonEmpty()$ \LogicAnd $\abs{K} < k$} {
            $u\gets H.\extract()$\;
            \ForAll(\tcp*[f]{Recall $\delta \leq k$}){edge $(u, v)$}{
                % \If {\Relax{$u, v$} \LogicAnd $d[u] + w_{uv} < \expectB$} {
                \If {\Relax{$u, v, \expectB$}} {
                    Add vertex $v$ and edge $(u, v)$ into $K$\;
                    \uIf {$v$ is in some existing $\bar{F}_j$}{
                        $\bar{F}_j \gets \bar{F}_j \cup K$\;
                        Quit the \KeyWhile-loop, skip $x$ and proceed to the next vertex in $S$\;
                    }\uElseIf {$v$ is in $H$ } {
                        Remove old incoming edge of $v$ from $K$\;
                        $H.\decrkey(v, d[v])$
                    }\lElse{
                        $H.\ins(v, d[v])$
                    }
                }
            }
        }
        \lIf {$\abs{K} \geq k$ } {
            Report $K$ as a new directed tree $\bar{F}_j$
        }\lElse{
            Report $K$ as a new arborescence $W_j$
        }
    }
    $W \gets \bigcup_{j}W_j$; $Q\gets \{\text{roots of }W_j\}$\tcp*[r]{$\abs{W} \leq O(k\abs{Q})$}
    Partition $\{\bar{F}_j\}$ into $\Theta(k)$ sized edge-disjoint sub-trees $\{F_j\}_{j=1}^{p}$ \tcp*[r]{See \cref{lemma:partition}}
    \For {$j \gets 1$ \KeyTo $p$} {
        $P_j \gets \{x \in (S\setminus Q) \cap F_j: \forall j' < j, x\notin F_{j'}\}$\tcp*[r]{$j$ minimality for disjointness}
    }
    \Return $\{P_j\}_{j=1}^{p}, Q, W$%\tcp*[r]{Suppose there are $p$ such non-empty $P_j$'s}
}
\end{algorithm}

\begin{lemma}[Finding Pivots]
\label{lemma:find-pivots}
%Under the same assumption with \cref{lemma:BMSSP} (except $l$): given an upper bound $\expectB$ and a vertex set $S$ of size at most $t^22^{lt}$,
Suppose $\frontierPair{\emptyset, S}$ is a frontier for $\target = \target{\expectB, S}$,
\FindPivots{$\expectB, S$} (\cref{alg:find-pivots})
%is a sub-routine of \cref{alg:BMSSP}, and it
returns subsets $\{P_j\}_{j=1}^{p}$ of $S$ each of size $O(k)$, another subset $Q\subseteq S$ and a set $W\subseteq \target$ of size $O(k\abs{Q})$, in $O((p+\abs{Q})k(\delta + \log k))=O((p+\abs{Q})k\log k)$ time.
After running \cref{alg:find-pivots}, $\frontierPair*{W, \bigcup_{j=1}^{p}P_j}$ is a frontier for $\target$.
\end{lemma}
\begin{proof}
We only need to prove that $\frontierPair*{W, \bigcup_{j=1}^{p}P_j}$ is a frontier for $\target$.

For any $v\in\target$, its shortest path visits some complete $u\in S$.
If $u$ is contained in some $P_j$, we are already done;
otherwise the sub-tree searched from $u$ contains less than $k$ vertices. Since $v$ lies in such a sub-tree, $v$ is included into $W$ and becomes complete. 
% One thing to note is that the first statement does not mean all vertices in $W_j$ are complete, but if a vertex in $W_j$ is incomplete, it satisfies the second statement.

As for the running time, the local Dijkstra search takes $O(1)$ time per processed edge, and $O(\log k)$ time per processed vertex as promised by Fibonacci heap~\cite{FT87};
by \cref{lemma:partition}, the partition step takes $O(1)$ time per edge.
Every processed edge is incident to a vertex in some $\bar{F}_j$ or some $W_j$, so there are most $O((p+\abs{Q})k\delta)$ of them.
Every processed vertex evokes heap insertion only once as a vertex of some $\bar{F_j}$, or possibly many times as a vertex in several $W_j$'s, which results in $O((p+\abs{Q})k\log k)$ time. Recall that $\delta \leq \log k$, so the total time is $O((p+\abs{Q})k\log k)$.
\end{proof}

\begin{remark}
\label{remark:find-pivots}
We also note that: 
\begin{itemize}
    \item Each $P_j$ is connected by a directed tree $F_j$, where $F_j$'s are edge-disjoint, and vertices of $F_j$ are in $\target$.
    \item $\{P_j\}_{j=1}^{p}$ and $Q$ are disjoint, and their union is $S$. Moreover, $p \leq \min\{\abs{S}, \abs*{\target}/k\}$.
    \item $W$ is a union of $\abs{Q}$ arborescences $\{W_j\}_{j=1}^{\abs{Q}}$; $Q$ are their roots.
    \item However, $\bigcup_{j=1}^{p}F_j$ and the arborescences in $W$ are not necessarily disjoint.
\end{itemize}

% Moreover, if $q_j\in Q$ is the root of $W_j$ and is complete at the beginning, then after running \cref{alg:find-pivots}, every vertex $v$ with $\distV(v) < \expectB$ whose shortest path visits $v$ is in $W_j$ and is complete. As a result, denote $Q^* = \{q\in Q: q\text{ is complete in the beginning}\}$, then $\target{\expectB, Q^*} \subseteq W$ and $\target{\expectB, Q^*}$ is complete.
% % i.e., $\target{\expectB, q_j}$ is a subset of $W_j$ and is complete.

\end{remark}

\subsection{Data Structure}

% (rearranging them into $M$-sized unsorted blocks),
\cite{SSSP25} designed a data structure to support the operations required for partially sorting the ``frontiers''.
In this section we present a simpler data structure with a slight improvement.
% By exploiting the remaining content in the data structure from sub-calls,
We implement by a self-balanced binary search tree (BST) of $M$-sized blocks: blocks sorted by the BST, but items inside a block unsorted.
Requiring $\log(N/M)\leq M$, an $O(\log (N/M))$ time BST operation is amortized to $O(1)$ per item in blocks, which enables us to perform operations except insertions in amortized linear time.

\begin{lemma}[Data Structure]
\label{lemma:data-structure}
Given at most $N$ key/value pairs involved, a parameter $M$, and an upper bound $\expectB$, there exists a data structure $\ds$ that supports the following operations:
\begin{description}
    \item[Insert] Insert a key/value pair. If the key already exists, keep the one with smaller value.
    \item[Merge] For another data structure $\ds'$ with smaller parameter $M'(< M/3)$ as specified in this lemma, such that all the values in $\ds'$ are smaller than all the values in $\ds$, insert all pairs of $\ds'$ into $\ds$ (keep the smaller value for duplicate keys).
    % If $\ds'$ is empty, do nothing; otherwise, we require $\abs{\ds'} = \Omega(M)$.
    \item[Pull] Return a subset $S'$ of keys where $\abs{S'}\leq M$ associated with the smallest $\abs{S'}$ values and an upper bound $x$ that separates $S'$ from the remaining values in $\ds$. Specifically, if there are no remaining values, $x = \expectB$; otherwise, $\abs{S'} = M$.
\end{description}
% \begin{description}
%     \item[Insert] Insert a key/value pair in $O(\log(N/M))$ amortized time. If the key already exists, keep the smaller.
%     \item[Merge] For another data structure $\ds'$ with smaller parameter $M'(< M/3)$ as specified in this lemma, such that all the values in $\ds'$ are smaller than all the values in $\ds$, insert all elements of $\ds'$ into $\ds$ (keep the smaller value if there are duplicate keys) in $O(\abs{\ds'})$ amortized time.    
%     \item[Pull] Return a subset $S'$ of keys where $\abs{S'}\leq M$ associated with the smallest $\abs{S'}$ values and an upper bound $x$ that separates $S'$ from the remaining values in $\ds$, in $O(\abs{S'})$ amortized time. Specifically, if there are no remaining values, $x$ should be $\expectB$. Otherwise, we have $\abs{S'} = M$.
% \end{description}

If $M = 1$, \ins\ and \pull\ take $O(\log N)$ time with \mergeds\ unsupported, for base case only (\cref{subsec:base-case}).

If $M > 1$, we require $M\geq \log(N / M)$. \ins\ takes $O(\log (N/M))$ amortized time; other operations are linear: \mergeds\ takes $O(\abs{\ds'})$ time and \pull\ takes $O(\abs{S'})$ time.
\end{lemma}

The proof of \cref{lemma:data-structure} is given in \cref{subsec:proof-of-data-structure}.

\subsection{Bounded Multi-Source Shortest Paths}

In this section we are ready to present our main \BMSSP algorithm, see~\cref{alg:BMSSP}.

\begin{algorithm}[p]
\caption{BMSSP}
\label{alg:BMSSP}
\AlgInput{Upper bound $\expectB$, vertex set $S$($\abs{S}\leq t^22^{lt}$), layer $l$}
\AlgOutput{New bound $\actualB (\leq \expectB)$; vertex set $\actualU(\abs{\actualU} = O(t^32^{lt}))$; data structure $\ds$}
\Alg{\BMSSP{upper bound $\expectB$, vertex set $S$, layer $l$}}{
    \lIf{$l = 0$}{
        Operate and return as \cref{alg:base-case} (Base Case)
    }
    Initialize data structure $\ds$ parameterized by $M = t2^{(l-1) t}$\;
    $\{P_j\}_{j=1}^{p}, Q, W \gets$ \FindPivots{$\expectB, S$}\label{line:bmssp:find-pivot-call}\;
    % \ForAll{$1\leq j \leq p$}{
    \For{$j \gets 1$ \KeyTo $p$} {
        $p_j \gets \arg\min_{x\in P_j}d[x]$\;
        $\ds.\ins(\pair{p_j, d[p_j]})$ \label{line:bmssp:initial-pivot-insert}
    }
    $\actualB \gets \actualB_{0} \bydef d_{\expectB}[\{p_j: 1\leq j\leq p\}]$;
    $\actualU\gets \emptyset$;
    $i \gets 1$;
    $\dirtyPj \gets \emptyset$
    \label{line:bmssp:define-B-zero}\;
    \While {$\abs{\actualU} \leq t^3 2^{lt}$ \LogicAnd $\ds$ is non-empty\label{line:bmssp:while-start}} {
        $S_i, \expectB_{i} \gets \ds.\pull()$ \label{line:bmssp:pull}\;
        \ForAll (\tcp*[f]{$\abs{S_i}$ expands by at most $k$ times.}) {$x \in S_i$\ } {
            \If {$x = p_j$\ for some $j$} {
                $S_i \gets S_i \cup \{v \in P_j: d[v] < \expectB_{i}\}$ \label{line:bmssp:pick-from-Si}           
            }
        }
        % \ForAll (\tcp*[f]{$\abs{S_i}$ expands by at most $k$ times.}) {$p_j \in S_i$ for some $j$ } {
        % $S_i \gets S_i \cup \{v \in P_j: d[v] < \expectB_{i}\}$ \label{line:bmssp:pick-from-Si}\;
        % }
        % $S_i \gets S_i \cup \{x \in P_j: d[x] < \expectB_{i} \text{\ and\ }  p_j \in S_i\}$\label{line:bmssp:pick-from-Si} \;
        $\actualB_{i}, \actualU_{i}, \ds_i \gets $ \BMSSP{$\expectB_{i}, S_i, l - 1$}\label{line:bmssp:sub-call}\;
        $\ds.\mergeds(\ds_i)$ \label{line:bmssp:merge}\;
        \ForAll {$u \in \actualU_{i}$} {
            \If {$u \in P_j$ for some $j$} {
                Remove $u$ from $P_j$\label{line:bmssp:remove-u-from-pj}\;
                \If{$u = p_j$ \LogicAnd $P_j\neq \emptyset$}{Add $j$ into $\dirtyPj$} \label{line:bmssp:pivot-remove-add-j}
            }
        }
        \ForAll {$u \in \actualU_{i}$} {
            \ForAll (\tcp*[f]{$u$ is complete, $d[u]=\distV(u)$}) {edge $(u, v)$}{
                 \If {\Relax{$u, v, \expectB$} \LogicAnd $d[u] + w_{uv} \geq \expectB_{i}$}{ \label{line:bmssp:edge-direct-check}
                    $\ds.\ins(\pair{v, d[v]})$ \label{line:bmssp:edge-direct-insert}\;
                    % Insert $\pair{v, d[v]}$ into $\ds$ \label{line:bmssp:edge-direct-insert}\;
                    \If {$v\in P_j$ for some $j$ \LogicAnd $j\notin \dirtyPj$ \LogicAnd $d[p_j] > d[v]$ } {
                        $p_j \gets v$\label{line:bmssp:update-pj-direct-insert}\;
                    }
                }
            }
        }
        \ForAll{$j\in \dirtyPj$}{
                $p_j \gets \arg\min_{x\in P_j}d[x]$\;
                $\ds.\ins(\pair{p_j, d[p_j]})$ \label{line:bmssp:tree-pivot-insert}
                % Insert $\pair{p_j, d[p_j]}$ into $\ds$ \label{line:bmssp:tree-pivot-insert}
        }
        $\actualB \gets \actualB_{i}$;
        $\actualU\gets \actualU \cup \actualU_{i}$;
        $i \gets i + 1$;
        $\dirtyPj \gets \emptyset$
        \label{line:bmssp:while-end}\;
    }
    \ForAll {$x\in S$ \label{line:bmssp:quit-while-loop}} {
        \If {$d[x] \in [\actualB, \expectB)$ } {
            $\ds.\ins(\pair{x, d[x]})$\label{line:bmssp:add-back-Si}\;
        }
    }    
    % $W' \gets \{x\in W \setminus \actualU: d[x] < \actualB\}$\label{line:bmssp:define-w-prime}\;
    \ForAll{$u\in W' \bydef \{x\in W \setminus \actualU: d[x] < \actualB\}$ } {\label{line:bmssp:define-w-prime}
        \ForAll(\tcp*[f]{$u$ is complete, $d[u]=\distV(u)$}){edge $(u, v)$ } {
            \If {\Relax{$u, v, \expectB$} \LogicAnd $d[u] + w_{uv} \geq \actualB$ \label{line:bmssp:check-w} }  { 
                $\ds.\ins(\pair{v, d[v]})$ \label{line:bmssp:insert-w}\;
                % Insert $\pair{v, d[v]}$ into $\ds$\label{line:bmssp:insert-w}\;
            }
        }
    }
    \Return $\actualB$, $\actualU\gets \actualU\cup W'$, $\ds$ \label{line:bmssp:return}\;
}
% \Fn (\tcp*[f]{$O(1)$ time to update $p_j$ to minimum in $P_j$.}) {\UpdatePj{vertex $u$}}{
%     \lIf{$u \in P_j$ for some $j$ \LogicAnd ($p_j \notin P_j$ \LogicOr $d[p_j] > d[u]$)} {
%         $p_j\gets u$
%     }
% }
\end{algorithm}

%\cref{alg:BMSSP} shares many similarities with \cite{SSSP25} and \cite{SSBP18}, acting as kind of a mixture, and some details are specified to fit the situation of this single-source shortest path problem.

The base case ($l=0$) can be solved using Dijkstra's algorithm as in \cref{subsec:base-case}. Now we assume $l > 0$.

We first initialize data structure $\ds$ with parameter $M = t2^{(l-1)t}$ to hold the vertices we will scan.

Then we run \FindPivots of \cref{lemma:find-pivots} to reorganize $S$ into disjoint subsets $\{P_j\}_{j=1}^{p}$ and $Q$ of $S$, with a set $W$.
For each $P_j$ we find $p_j = \arg\min_{x\in P_j}d[x]$ the minimum vertex in it.
We will update $p_j$ accordingly such that $p_j$ always lower bounds necessary
\footnote{By ``necessary'' we mean a vertex in $P_j$ whose shortest path does not visit any complete vertex in $\ds$;
because if it does visit, such a vertex need not to be tracked.
For better readability,
we recommend that readers first form a high-level picture of the algorithm and then return to this definition.
A formal proof that the algorithm preserves this property appears in \cref{lemma:bmssp-correctness-frontier}.
}
vertices in $P_j$.

Set $\actualB \gets \actualB_{0} \bydef \min\{\expectB, \{d[p_j]: 1\leq j\leq p\}\}$;
$\actualU\gets \emptyset$.
$\actualB$ actually marks the progress of our algorithm. $\actualU$ will accumulate batches of complete vertices.
Set $\dirtyPj\gets \emptyset$; $\dirtyPj$ holds the set of $j$ such that $p_j$ needs to be re-selected.
The rest of algorithm repeats the following steps for multiple rounds until $\ds$ becomes empty or $\abs{\actualU} > t^32^{lt}$, where during the $i-$th iteration ($i\geq 1$), we:
\begin{itemize}
    \item Pull from $\ds$ a subset $S_i$ of keys with smallest values in $\ds$ upper bounded by $\expectB_{i}$ (also a lower bound for remaining vertices in $\ds$).
    Scan $x\in S_i$; if $x=p_j$ for some $j$, add $\{v\in P_j: d[v] < \expectB_{i} \}$ into $S_i$ (by brute-force search);
    \item Recursively call \BMSSP{$\expectB_{i}, S_i, l-1$} and collect its output $\actualB_{i}, \actualU_{i}, \ds_i$;
    \item Merge vertices of $\ds_i$ into $\ds$;
    % (for each $u\in \ds_i$, also update $p_j$ if it is in $P_j$ by \UpdatePj);
    \item For each $u\in \actualU_{i}$, if $u \in P_j$ for some $j$, remove $u$ from $P_j$;
    and if $u = p_j$, also add $j$ into $\dirtyPj$.
    \item For each $u\in \actualU_i$, relax every edge $(u, v)$ such that $d[u] + w_{uv} \in [\expectB_{i}, \expectB)$;
    if the relaxation is valid, insert $\pair{v, d[v]}$ into $\ds$;
    if $v\in P_j$ for some $j\notin \dirtyPj$,
    also update $p_j$ to the smaller one between $v$ and $p_j$;
    \item For each non-empty $P_j$ from $j\in \dirtyPj$, re-select $p_j$ to the minimum vertex in $P_j$ by brute-force search and insert it into $\ds$;
    \item Update $\actualU$ to include $\actualU_{i}$
    and $\actualB\gets \actualB_{i}$;
    Reset $\dirtyPj \gets \emptyset$;
    \item If $\ds$ becomes empty, then this is a full execution and we quit the loop;
    \item If $\abs{\actualU} > t^32^{lt}$,
    this is a partial execution where we encountered a large workload. Also quit the loop.
    % \item Otherwise if $\ds$ is non-empty and $\abs{\actualU} \leq t^32^{lt}$, return to the first step.
\end{itemize}

After quitting the loop, we operate the following finalizing operations:
\begin{itemize}
    \item For every $x\in S$ such that $d[x] \in [\actualB, \expectB)$, insert $\pair{x, d[x]}$ into $\ds$.
    \item For every edge $(u, v)$ from every $u\in W' \bydef \{ x\in W \setminus \actualU: d[x] < \actualB\}$ such that $d[u] + w_{uv} \in [\actualB, \expectB)$, relax it;
    if the relaxation is valid, insert $\pair{v, d[v]}$ into $\ds$.
    % \item For every such $u$ above, i.e., for every $u\in W'$, add it into $\actualU$.
    \item $\actualU \gets \actualU \cup W'$, i.e., add every vertex $u\in W'$ into $U$.
\end{itemize}

% Finally, before termination, we update $\actualU$ to include every vertex $x$ in any $W_j$ returned by \FindPivots with $d[x] < \actualB$.

\subsection{Observations and Discussions on the Algorithm}
\label{subsec:obserevation}

We first present some informal explanations to give reader an intuitive overview on the algorithm, and in subsequent subsections we will provide more formal proofs on correctness and running time.

\paragraph{Recursion.} We first think of the recursion tree $\mathcal{T}$ of \cref{alg:BMSSP} where each node $X$ denotes a call of \BMSSP algorithm. We subscript $X$ to all the parameters to denote those used in the call of $X$: $l_X, \expectB_{X}, S_X$ for input, $\actualB_X, \actualU_{X}, \ds_X$ for output, $\{P_{X, j}\}_{j=1}^{p_X}, Q_X, W_X$ for output of \FindPivots (\cref{line:bmssp:find-pivot-call}). Also denote $W'_X = \{x\in W_X \setminus \actualU_{X}: d(x) < \actualB_X\}$, $P_X = \bigcup_{j=1}^{p_X}P_{X, j}$, and $\target_{X} = \target{\expectB_{X}, S_X}$ as above.
Then we have the following properties:
\begin{enumerate}
    \item In the root $R$ of $\mathcal{T}$, $\actualU_{R}=V$;
    \item Think of the construction of $S_X$ in the caller $Z$ of $X$: when pulled (\cref{line:bmssp:pull}), $\abs{S_{X}} \leq M_{Z} = t2^{(l_{Z}-1)t} = t2^{l_{X}t}$.
    Adding $\{v\in P_{Z,j}: d[v] < \expectB_{X}, p_{Z, j} \in S_X\}$%\footnote{$S_X$ is exactly $S_i$ as $X$ denotes the execution of this round.} 
    into $S_{X}$ (\cref{line:bmssp:pick-from-Si}) expands $\abs{S_X}$ by at most $O(k)$ times.

    Therefore $\abs{S_X} \leq t^22^{l_{X}t}$, so the depth of tree $\mathcal{T}$ is at most $(\log n)/t$;
    \item We break when $\abs{\actualU_{X}} > t^32^{l_Xt} \geq t\abs{S_X}$. Intuitively, $\abs{\actualU_{X}}$ grow slowly enough so we still have $\abs{\actualU_{X}} = O(t^32^{l_Xt})$ (see~\cref{lemma:bmssp-correctness-size}).
    \item If $X$ is a full execution, then $\target_{X} = \actualU_{X}$. If $X$ is a partial execution, then $\abs*{\actualU_{X}} > t\abs{S_X}$.
    In both cases,
    % $\min\{\abs*{\target_{X}}, t\abs{S}\} \leq \abs{\actualU_X}$.
    by \cref{remark:find-pivots}, $p_X \leq \min\{\abs*{\target_X}/k, \abs{S_X}\}\leq \abs{\actualU_{X}} / k$. \label{item:recursion:both-size-constraint}
    \item \label{item:recursion:ui-disjoint} Suppose $Y_1, Y_2, \cdots, Y_f$ are the sub-calls invoked by $X$, i.e., the children of $X$ in $\mathcal{T}$. Then we have:
    \begin{itemize}
        \item $\actualU_{Y_1}, \actualU_{Y_2}, \cdots, \actualU_{Y_f}$ are disjoint. For $i < j$, $\distV()$ of vertices in $\actualU_{Y_i}$ are smaller than that of $\actualU_{Y_j}$.
        \item $\actualU_{X}$ is a disjoint union of $W'_X$ and $\actualU_{Y_1}, \actualU_{Y_2}, \cdots \actualU_{Y_f}$.
        \item The $\actualU_{X}$'s for all nodes at the same layer in $\mathcal{T}$ are disjoint, and total size of all these $\abs{\actualU_{X}}$ in each layer is $O(\abs{V})$.
    \end{itemize}
    \item For any vertex $u$, there is only at most one call in each layer whose $\actualU$ contains it.
    Thus, the calls $X$ whose $\actualU_{X}$ contains $u$ forms a path in $\mathcal{T}$: $R = X_0, X_1, \cdots, X_q$ such that: \label{item:recursion:relax-unique}
    \begin{itemize}
        \item For $0 \leq i < q$, $X_{i+1}$ is a sub-call of $X_{i}$;
        \item For $0 \leq i < q$, $u\in \actualU_{X_i}$; relaxations from $u$ (\cref{line:bmssp:edge-direct-insert}) lie in interval $[\expectB_{X_{i+1}}, \expectB_{X_i})$;
        \item For the last call $X_q$ containing $u$, either $u\in W'_{X_q}$ and relaxations from $u$ (\cref{line:bmssp:insert-w}) lie in interval $[\actualB_{X_q}, \expectB_{X_q})$;
        or $u\in \actualU_{X_q}$ where $X_q$ is a base case call and relaxations from $u$ (\cref{line:base:relax}) lie in interval $[0, \expectB_{X_q})$;
        \item So the intervals for $X_0, \cdots, X_q$ are disjoint, and any edge can be validly relaxed at most once, in \cref{line:bmssp:edge-direct-insert}, \cref{line:bmssp:insert-w}, or a base case call. %, among all calls in $\mathcal{T}$.
    \end{itemize}

\end{enumerate}

\paragraph{Correctness.} Base case is solved using Dijkstra's algorithm and is correct. Below we assume $l > 0$.

The framework of \cref{para:frontier-and-dijkstra-like-alg-framework} helps understand how our algorithm works and verify correctness.
In a call of \BMSSP, $\target = \target{\expectB, S}$ is the set that contains all vertex $v$ with $\distV(v) < \expectB$ whose shortest path visits some vertex in $S$, and is our expected target. Then \BMSSP should return $\actualU = \target$ in a successful execution, or $\actualU = \{u\in \target: \distV(u) < \actualB\}$ in a partial execution, with all vertices in $\actualU$ complete.

At the beginning, we reorganize $S$ as $\Theta(k)-$sized chunks $\{P_{X, j}\}_{j=1}^{p_X}$ to insert into $\ds$, with some vertices complete and added into $W$. \cref{lemma:find-pivots} ensures that the shortest path of remaining incomplete vertices in $\target$ visit some complete vertex in $P$, i.e., we push the frontier to $\frontierPair{W, P}$.

We do not insert all of $P$, but only the minimum vertex (``pivot'') of each $P_j$ , i.e., $p_j = \arg\min_{x\in P_j}d[x]$ into $\ds$.
Then the shortest paths of remaining incomplete vertices in $\target$ visit some complete vertex in $\ds \cup P$.
For any bound $\expectB_{i} \leq \expectB$, if $\distV(v) < \expectB_{i}$,
the shortest path of $v$ must also visit some complete vertex in $u\in \ds \cup P$ with $\distV(u) < \expectB_{i}$.
When we pull $S_i$ a subset from $\ds$, the pivots $p_j\in S_i$ helps augment $S$ to contain necessary vertices bounded by $\expectB_{i}$ from $\ds \cup P$.
Then we call sub-procedure \BMSSP on $\expectB_{i}$ and $S_i$.

By the inductive hypothesis, each recursive call on \cref{line:bmssp:sub-call} of \cref{alg:BMSSP} returns $\actualU_{i}$ and $\ds_i$, vertices in $\actualU_{i}$ are complete,
% and the shortest path of remaining incomplete vertices in $\target_i = \target{\expectB_{i}, S_i}$ visit some complete vertex in $\ds_i$.
% $\ds_i$ holds $\actualU_{i}$'s out-neighbors $x$ with $d[x] \in [\actualB_{i}, \expectB_{i})$ and remaining vertices in $P_{j}$.
and $\ds_i$ contains vertices remaining in $S_i$, and vertices validly relaxed from $U_i$ with distance less than $\expectB_{i}$.
What we do actually matches \cref{para:frontier-and-dijkstra-like-alg-framework}:
relax edge $(u, v)$ from $u\in U_i$, if $\actualB_i \leq\distV(u) + w_{uv} < \expectB_{i}$, $v$ is already in $\ds_{i}$; otherwise, $v$ is relaxed and inserted to $\ds$ directly.
% add $\ds_i$ into $\ds$ and relaxes edges from $u\in \actualU_{i}$ and inserts all the updated out-neighbors $x$ with $d[x] \in [\expectB_{i}, \expectB)$ into $\ds$,
Once again remaining incomplete vertices in $\target$ now visits some complete vertex in $\ds \cup P$.

% To preserve minimality of $p_j$,
If $p_j$ is in $U_i$ and removed,
we re-select a new minimum $p_j$ and insert it into $\ds$.
It becomes complicated when $p_j$ is not removed while still $d[v]$ gets decreased for some other $v\in P_j$, because we cannot afford to re-select a new $p_j$.
In this case we update $p_j$ to lower bound only a subset $\hat{P}_j \subseteq P_j$, the set of vertices in $P_j$ whose shortest path do not visit any complete vertex in $\ds$ --- we do not explicitly compute $\hat{P}_j$, merely introducing it for explanation.
Because we may rely on $\ds$, we do not need to keep track of vertices not in $\hat{P}_j$. This observation will be explained further in \cref{lemma:bmssp-correctness-frontier}.

After we quit the loop, if this is a full execution ($\actualB = \expectB$), we return $\ds = \emptyset$. If this is a partial execution, $\ds$ already includes all vertex $v$ relaxed from all $u\in\actualU_i$ which satisfies $\actualB \leq\distV(u) + w_{uv} < \expectB$. Then we first add back remaining vertices in $S$, whose distance labels are in $[\actualB, \expectB)$; they are premature vertices and need to be added back to frontier again in the upper layer. We also need to relax from $W' = \{x\in W \setminus \actualU: d[x] < \actualB \}$ which includes vertices whose shortest paths visits complete vertices in $Q$.

%We also relax from $W' = \{x\in W \setminus \actualU: d[x] < \actualB \}$.
%Here we explain why we find $W$ from the very beginning, but only process it at the end.
%Recall in \FindPivots, $W$ is generated by local Dijkstra searches which fail to reach $k$ vertices from their roots in $Q$.
% For vertices whose shortest path visits complete vertices in $Q$, they are already complete (denote by $W_{com}$); for other vertices, they may be reached by complete or incomplete vertices in $Q$ but their shortest path 
%For vertices reached by complete vertices in $Q$, they are already complete (denote by $W_{com}$);
%for those reached by incomplete vertices in $Q$, their shortest paths visit complete vertices in $P$, and they will be updated by the major while-loop of our algorithm (denote by $W_{inc}$).
%The problem is, we do not know which vertices in $Q$ are complete immediately after \FindPivots call.
%At the end when we quit the while-loop, we know $W_{inc}$ vertices whose distance label are less than $\actualB$ are already complete;
%and $W_{com}$ vertices are complete by \FindPivots from the beginning.
%Therefore, without knowing whether a vertex in $Q$ is complete,
%we know $\{x\in W: d[x] < \actualB\}$ is complete.
%Then we use it to update the frontier as in \cref{para:frontier-and-dijkstra-like-alg-framework}; excluding $\actualU$ because they have been updated in the while-loop.

\paragraph{Running time} The running time is dominated by calls of \FindPivots, the overheads inside the data structures $\ds$, selecting $p_j = \arg\min_{x\in P_j}d[x]$ from $P_j$'s and picking $\{x\in P_j: d[x] < \expectB_{i}\}$ for each $p_j \in S_i$. First we make an observation:
\begin{observation}\label{ob:relaxation}
    Any edge $(u, v)$ can only enter the execution of \cref{line:bmssp:edge-direct-insert} once in the whole algorithm.
    % Any edge $(u,v)$ can only be relaxed and $v$ added into $\ds$ on \cref{line:bmssp:edge-direct-insert} once in the whole algorithm.
\end{observation}
\begin{proof}
    This can only happen when $u$ becomes complete in a call $X$ which is a child of current call $Y$ in $\mathcal{T}$, so there is at most once in one layer. Other call $Z$ where $U_Z$ contains $u$ must be ancestor or descendant of $X$. Also in the current call $Y$ we need $\expectB_X \leq d[u]+w_{uv}< \expectB_Y$ to make \cref{line:bmssp:edge-direct-insert} happen, but for descendant $Z$ of $X$ (including $X$ itself) in $\mathcal{T}$, $\expectB_Z\leq \expectB_X$, so it is impossible that $d[u]+w_{uv}< \expectB_Z$. So \cref{line:bmssp:edge-direct-insert} will not happen in lower layers, thus only once in the whole recursion.
\end{proof}

By \cref{item:recursion:both-size-constraint} and \cref{lemma:find-pivots} above, $p_X \leq \abs{\actualU_{X}}/k$, so the total running time of \FindPivots over one layer of $\mathcal{T}$ is $O(\abs{V}\log t + \sum_{\text{$X$ in one layer of $\mathcal{T}$}}\abs{Q_X}t)$, since $t=k\log t>k\log k$.
To estimate $\sum_{X\in \mathcal{T}}\abs{Q_X}$:
\begin{itemize}
    \item If $X$ is a partial execution, $\abs{Q_X} \leq \abs{S_X} < \abs{U_X} / t$.
    \item If $X$ is a full execution, for any $v\in Q_X$, $v\in\actualU_X$.
    For $v$'s latest relaxed edge $(u, v)$ at the beginning of call $X$ (Line~\ref{line:bmssp:edge-direct-insert}, \ref{line:bmssp:insert-w}):
    \begin{itemize}
        \item Since $v\in S_X$, $u$ is not in $\actualU_X$.
        \item If there is another full execution $Y$ with $v\in Q_Y$, $X$ and $Y$ is of ancestor-descendant relationship. 
        \item We do not put vertices in $Q_X$ in $\ds_X$ directly, thus if $v\in Q_Y$ for a descendant $Y$ of $X$ which is also a full execution, its latest relaxation edge must be different from $(u, v)$.
    \end{itemize}
    %but $u$ must be in $\actualU_Y$ where $Y$ is the parent of $X$. So $Y$ is like a least common ancestor (LCA) of $u,v$, so this can only happen once for each edge $(u,v)$.
    %Edge $(u, v)$ will never be relaxed (\cref{item:recursion:relax-unique}) again. $v$ becomes complete after $X$, and will not be added back to $\ds_{X}$ until the next relaxation in \cref{line:bmssp:edge-direct-insert} or \cref{line:bmssp:insert-w}.
    % in sub-call $X'$ under the subtree of $X$ in $\mathcal{T}$.
    Thus the number of appearances of $v$ in $Q_X$ in all full executions $X$ is bounded by $\delta$, the degree bound of $v$. Therefore $\sum_{X \in \mathcal{T}} \abs{Q_X} = O(\abs{V}(\delta+l/t))$.
\end{itemize}
The total time for \FindPivots is $O(\abs{V}l\log t+ t\sum_{X\in \mathcal{T}}\abs{Q_X}) = O(mt +  (m\log n\log t )/ (t\delta))$. (Recall the number of vertices is $O(m/\delta)$.)

For the data structure $\ds$ in a call of \BMSSP, the total running time is dominated by \ins\ as all other operations take linear time.
It is clear that vertices added into $\ds_X$ include a subset of $S_X$, and some out-going neighbors of $\actualU_{X}$, so the size of $\ds_X$ is bounded by $\abs{S_X}+\abs{\actualU_{X}}\delta \leq t^42^{lt}$, so each insertion takes time $O(t)$.

\begin{itemize}
    \item For pivots $p_j$ to $\ds$ in \cref{line:bmssp:initial-pivot-insert}, $p_X \leq \abs{\actualU_{X}}/k$, so the time is $O(\abs{\actualU_{X}}t/k)=O(\abs{\actualU_{X}}\log t)$. %the total time for nodes over one layer of $\mathcal{T}$ is $O(\frac{m}{\delta} \frac{t}{k} ) = O(\frac{m\log t}{\delta})$. Summing over all depths we get $O(m\frac{\log n\log t}{\delta t})$.
    \item For vertices inserted into $\ds$ in \cref{line:bmssp:add-back-Si}, they are non-zero only when $X$ is a partial execution ($\actualB_{X} < \expectB_{X}$), such that $\abs{S_X} \leq \abs{\actualU_{X}} / t$, so the time is $O(t\abs{S_X}) \leq O(\abs{\actualU_X})$.
    %The sum over one layer of $\mathcal{T}$ is $O(m/\delta)$ and summing over all depths we get $O(m\frac{\log n}{\delta t})$.
    \item By \cref{ob:relaxation}, a vertex $v$ can only be inserted to $\ds$ by edge $(u, v)$ in \cref{line:bmssp:edge-direct-insert} once. Also a vertex $u$ can only be in $W'$ once as $W'\subseteq \actualU\setminus(\bigcup \actualU_i)$, so each edge $(u,v)$ can only be relaxed once in \cref{line:bmssp:insert-w}.
    
    %inserted to $\ds$ by edge $(u, v)$ in \cref{line:bmssp:edge-direct-insert} or \cref{line:bmssp:insert-w} in some call $X$, $(u, v)$ is relaxed at most once among all calls. In total the time is $O(mt)$.
\end{itemize}
Therefore, the total time for \cref{line:bmssp:initial-pivot-insert}, \cref{line:bmssp:add-back-Si}, \cref{line:bmssp:edge-direct-insert} and \cref{line:bmssp:insert-w} is $O(mt+(m \log t\log n)/ (t \delta))$.

Next we consider the time of picking (\cref{line:bmssp:pick-from-Si}), re-selection and insertion (\cref{line:bmssp:tree-pivot-insert}) of $X$. Similar to \cref{item:recursion:ui-disjoint}, suppose $Y_1, Y_2, \cdots, Y_f$ are sub-calls of $X$.

\begin{itemize}
    \item Picking $\{x\in P_j: d[x] < B_i\}$ (\cref{line:bmssp:pick-from-Si}) takes time $O(k)$ per $p_j\in S_{Y_i}$.

\begin{itemize}
    \item If $Y_i$ is a partial execution, time is bounded by $O(k\abs{S_{Y_i}}) = O(\abs{\actualU_{Y_i}})$;
    \item If $Y_i$ is a full execution, for any $p_j\in S_{Y_i}$, $p_j\in \actualU_{Y_i}$ and is removed from $P_j$, so $P_j$ needs to re-select and insert a new $p_j$ which takes $O(t)$ time.
    % greater than $O(k)$,
    We amortize this $O(k)$ to that operation as $k\leq t$. (If $P_j$ becomes empty, this happens at most one time for each $P_j$; these contribute to at most $O(kp)=O(\abs{\actualU_{X}})$.)
\end{itemize}
So this step takes time $O(\abs{\actualU_{X}})$ for $X$, and summing over all calls we get $O((m \log n)/(t\delta))$.

    \item Re-selection and insertion (\cref{line:bmssp:tree-pivot-insert}) take $O(t)$ time per $p_j$. For a call $X$,
\begin{itemize}
    \item If $X$ is a partial execution, time is bounded by $O(\abs{S_X}t)=O(\abs{\actualU_X})$.
    \item
    If $X$ is a full execution,
    % As $pt \leq \abs{\actualU} t/k = O(\abs{\actualU_{X}} \log t)$, we can directly afford $O(1)$ re-selection and insertions (\cref{line:bmssp:tree-pivot-insert}) for each $P_{X, j}$.
    for each re-selection and insertion (\cref{line:bmssp:tree-pivot-insert}) of each $P_{X, j}$, we want to establish an injection from them to one tree vertex of $F_{X, j}$ in each $U_{Y_i}$.
    The number of tree vertices in different $U_{Y_i}$'s is no more than one plus the number of ``cross-level'' edges in $F_{X, j}$, which are the edges connecting two different $\actualU_{Y_i}$'s, which are not used by full sub-calls (full $Y_i$ only uses edges connecting two vertices within $\actualU_{Y_i}$), nor callers of $X$ ($X$ is a full execution), such that these edges are globally uniquely consumed by $X$.

    % For each $P_{X, j}$, consider each $p_{X, j}$:
    % it is itself a tree vertex of $F_{X, j}$ and may belong to some $U_{Y_i}$, but this is not an injection, because $p_{X, j}$ may be swapped to other vertices (\cref{line:bmssp:update-pj-direct-insert}) such that two different $p_{X, j}$ may belong to the same $U_{Y_i}$.

    Instead of each $p_{X, j}$ in \cref{line:bmssp:tree-pivot-insert} itself, we think of its former pivot of $P_{X, j}$, $p'_{X, j}$, which is in $U_{Y_i}$ and was removed from $D$ in $i$-th round, so $j$ is put into $J$. Each $p_{X, j}$ uniquely corresponds to $p'_{X, j}$, and each $p'_{X, j}$ belongs to some unique $\actualU_{Y_i}$: there is at most one pivot $p'_{X, j}$ removed from $D$ in each $\actualU_{Y_i}$. 
    
    % For each $\actualU_{Y_i}$ where $P_{X, j}$ re-selected and inserted a new $p_{X, j}$ (\cref{line:bmssp:tree-pivot-insert}), there exists exactly one old $p_{X, j}'\in \actualU_{Y_i}$ which was removed (\cref{line:bmssp:pivot-remove-add-j}).

    % For each $P_{X, j}$, consider all the $p_{X, j}$'s which are re-selected and inserted into $D$ from $P_{X, j}$ (\cref{line:bmssp:tree-pivot-insert}), then each $\actualU_{Y_i}$ should contain at most one such $p_{X, j}$.
    
    % For each $p_{X, j}$ re-selected and inserted from $P_{X, j}$ over different rounds in $X$, it must correspond to exactly one old $p_{X, j}' \in \actualU_{Y_i}$ which was removed (\cref{line:bmssp:pivot-remove-add-j}).
    
    % Suppose there are $g_j$ such $p_{X, j}$ vertices for $P_{X, j}$ in the current call of $X$. They distributes in $g_j$ distinct $U_{Y_i}$'s. Then in each $F_j$, there exists at least $g_j-1$ ``cross-level'' edges, which are the edges connecting two different $U_{Y_i}$'s.
    % These edges are not used by sub-calls (full $Y_i$ only uses edges connecting two vertices within $U_{Y_i}$), nor callers of $X$ (because $X$ is a full execution).

    Therefore, we amortize the time of re-selection and insertion (\cref{line:bmssp:tree-pivot-insert}) for $P_{X, j}$ to edges, such that each edge is amortized at most once in the whole algorithm.
    (Except once for each $P_{X, j}$ which sums up to $O(p_{X}t)=O(\abs{\actualU_{X}}\log t)$.)
    %by some call among all calls.
\end{itemize}
Summing over all calls we get $O((m\log n \log t)/(t\delta)+mt)$. %\duan{originally $O(m(t+\frac{\log n\log t}{\delta t}))$.}
\end{itemize}

In total, the algorithm runs in $O(m(t + \frac{\log n\log t}{\delta t}))$ time, dominated by \FindPivots. 

\subsection{Correctness Analysis of BMSSP}

In this section we verify \cref{alg:BMSSP} fits the Dijkstra-like algorithm framework proposed in \cref{para:frontier-and-dijkstra-like-alg-framework}.

\cref{lemma:frontier-split} explains why we can invoke a recursive call on a smallest set extracted from the frontier.
\begin{lemma}[Frontier Split]
\label{lemma:frontier-split}
Suppose at some moment, $\frontierPair{X, Y}$ is a frontier for $\target = \target{\anyB, S}$. For a smaller $\anyB_{s} \leq \anyB$, suppose $\{x\in Y: d[x] < \anyB_{s}\} \subseteq Y_{s} \subseteq \target$.
Then $\frontierPair{\emptyset, Y_s}$ is a frontier for $\target{\anyB_{s}, Y_{s}}$.
\end{lemma}
\begin{proof}
For any $v\in \target{\anyB_{s}, Y_{s}}$, its shortest path visits some $y\in Y_s$.
If $y$ is complete, then we are done.
% If $v$ is complete, then $x$ is complete.
If $y$ is incomplete,
as $y\in Y_s \subseteq \target$, the shortest path of $y$, as a prefix of the shortest path of $v$, must visit some complete $y'\in Y$. Then the shortest path of $v$ visits $y'$ and $y' \in Y_s$ as
$d[y'] = \distV(y') \leq \distV(v) < \expectB_s$.
% as $\target{\anyB_{s}, Y_s} \subseteq \target$, its shortest path visits some complete $x'\in Y$. Since $d[x']=\distV(x') < \distV(v) < \anyB_{s}$, $x'\in Y_s$.
\end{proof}

We formally define the set $R$ mentioned in \cref{para:frontier-and-dijkstra-like-alg-framework}.
For a \emph{complete} set $Z$ and any bound $\anyB$, after relaxing the outgoing edges from $Z$ (not exceeding $\anyB$), 
denote $R(\anyB, Z)$ as the set of out-neighbors $v$ of $Z$ but not in $Z$ with distance less than $\anyB$, and relaxations from some vertex $u\in Z$ to $v$ is valid.
Equivalently, after relaxing outgoing edges from $Z$,
$R(\anyB, Z) = \{v: \exists u\in Z, (u, v)\in E, v\notin Z, d[v] = d[u] + w_{uv} < \anyB\}$.
% denote $R(\anyB, Z) = \{v \notin Z : \exists u\in Z, d[u] + w_{uv} = d[v] < \anyB\}$.
% When we mention this set $R$, there are two concepts: 
% (i) the set of vertices whose shortest path visits $Z$; especially, the last vertex on this path is a vertex in $Z$;
% (ii) the set of vertices, which can be validly relaxed from $Z$;
% The set of (i) are the vertices we need, and is a concept predetermined by the input, but we do not know which are they;
% The set of (ii) are the vertices we have, and we know (i) is included in (ii), and we just have the time to insert everything in (ii).

Denote $\hat{P}_j$ the set of vertex in $P_j$ whose shortest path does not visit any complete vertex in $\ds$.
% $\hat{P}_j$ could be informally interpreted as ``roots'' of vertices in $P_j$.
Also denote $P = \bigcup_{j=1}^{p}P_j$; $\hat{P} = \bigcup_{j=1}^{p}\hat{P}_j$.
We define them to be up-to-date through algorithm progress.
% They are merely for proofs but not actually implemented.
Clearly, any shortest path visits some complete vertex in $\ds \cup P$, if and only if it visits some complete vertex in $\ds \cup \hat{P}$.

We introduce $\hat{P}_j$ because for some $v\in P_j$,
after partial $i-$th sub-call,
$d[v]$ may decrease but $v$ is not in $\ds_i$ or $U_i$ such that we cannot maintain $p_j$ to minimum among $P_j$.
This is because in a partial execution as we will see in \cref{lemma:bmssp-correctness-frontier},
$\ds$ contains $S\setminus U$ and only a thin shell of out-going neighbors from $U$,
while vertices in $\target$ two or more hops away from $\actualU$ may be modified but not returned.
Fortunately, the shortest path of such vertices visit some complete vertex in $\ds$ such that we do not need to track them from $P_j$.

% Because $P$ and $P_j$ vary through the algorithm,
We denote $P_{ini}$ the initial value of $P$, i.e., $\bigcup_{j=1}^{p}P_j$ returned by \FindPivots.

% Why do we have to introduce $W'$? Because their sources $Q$

\begin{lemma}[Correctness of BMSSP, focusing on frontier and completeness]
\label{lemma:bmssp-correctness-frontier}
Given an upper bound $\expectB$, a vertex set $S$,
suppose $\frontierPair{\emptyset, S}$ is a frontier for $\target = \target{\expectB, S}$.
After running \cref{alg:BMSSP}:
\begin{enumerate}
    \item $\frontierPair*{\actualU, \ds}$ is a frontier for $\target$;
    \item $\max_{x\in \actualU} \distV(x) < \actualB \leq d_{B}[\ds]$, $\actualU = \target{\actualB, S}$ and is complete;
    \item $\ds$ contains $S\setminus \actualU$ and every out-neighbor $v$ of $\actualU$ but not in $\actualU$ validly relaxed from $u\in \actualU$ with $d[u] + w_{uv} = d[v] < \expectB$.
    Namely, $\ds \supseteq (S \setminus \actualU) \cup R(\expectB, \actualU)$.
\end{enumerate}
\end{lemma}
\begin{proof}
We prove by induction the following three propositions: at the beginning of $i-$th iteration (immediately before \cref{line:bmssp:pull}, or equivalently, immediately after \cref{line:bmssp:while-end}, at the end of $(i-1)-$th iteration):
\begin{enumerate}
    \item $\frontierPair*{U, \ds \cup P}$ is a frontier for $\target_{P}\bydef\target{\expectB, P_{ini}}$;
    \item $d_{B}[\ds \cup P] \geq B' (=B_{i-1}')$;
    \item For each non-empty $P_j$, $d[p_j] \leq d_{\expectB}[\hat{P}_j]$.
\end{enumerate}
Clearly they hold for the first iteration by \cref{lemma:find-pivots} and \cref{lemma:frontier-observation}. Now suppose they hold at the beginning of $i-$th iteration, we prove that they hold at the end of the $i-$th iteration.

We first show that: the prerequisite conditions of \cref{lemma:bmssp-correctness-frontier} are met for the $i-$th sub-call (\cref{line:bmssp:sub-call}).

Since $\frontierPair*{\actualU, \ds\cup P}$ is a frontier for $\target_{P}$, $\frontierPair*{\actualU, \ds \cup \hat{P}}$ is also a frontier for $\target_{P}$.
$\ds$ pulls $S_i$ bounded by $\expectB_i$ whereas $\expectB_i$ lower bounds remaining vertices in $\ds$ (\cref{line:bmssp:pull}).
Picking vertices in $P_j$ (\cref{line:bmssp:pick-from-Si}) adds every $v\in \hat{P}$ with $d[v] < \expectB_i$ into $S_i$, because $d[p_j]$ lower bounds each $\hat{P}_j$.
Thus, $S_i$ contains every $v\in \ds \cup \hat{P}$ with $d[v] < \expectB_i$.
By \cref{lemma:frontier-split}, $\frontierPair*{\emptyset, S_i}$ is a frontier for $\target{\expectB_i, S_i}$.
% the prerequisite conditions of \cref{lemma:bmssp-correctness-frontier} are met by the sub-call (\cref{line:bmssp:sub-call}).

Clearly base case satisfies \cref{lemma:bmssp-correctness-frontier}.
By induction, the sub-calls satisfy the conclusions of \cref{lemma:bmssp-correctness-frontier}.

% Hence, after $i-$th sub-call:
% $\actualU_i$ is complete;
% $\ds_i \supseteq (S_i \setminus \actualU_i) \cup R(\expectB_i, \actualU_i)$.

% Subsequent operations in $i-$th iteration is indeed using complete set $U_i$ to update its frontier $\frontierPair*{\actualU, \ds \cup P}$ to $\frontierPair*{\actualU \cup \actualU_{i}, ((\ds \cup P)\setminus \actualU_{i})\cup R(\expectB, \actualU_{i})}$ as described in \cref{para:frontier-and-dijkstra-like-alg-framework}:

Subsequent operations in $i-$th iteration is indeed using complete set $U_i$ to update its frontier $\frontierPair*{\actualU, \ds \cup P}$ as described in \cref{para:frontier-and-dijkstra-like-alg-framework}:
\begin{enumerate}
    \item $U$ obtains vertices in a complete set $\actualU_i$ (\cref{line:bmssp:while-end});
    \item $\ds \cup P$ loses \emph{only} vertices in $U_i$:
    \begin{itemize}
        \item $P$ loses vertices $U_i$ (\cref{line:bmssp:remove-u-from-pj});
        \item $\actualU_i$ is not in $\ds$, no need to remove: $\forall u\in \actualU_{i}$, $d[u] = \distV(u) < \actualB \leq d_{B}[\ds]$.  
        \item $S_i \setminus \actualU_i$ was removed from $\ds\cup P$, but merged back into $\ds$ again as $\ds_i\supseteq S_i\setminus U_i$ (\cref{line:bmssp:merge});
        % \item $\hat{P}$ might lose some vertex, the shortest paths of these vertices visits some complete vertex in $\ds$
    \end{itemize}
    \item $\ds \cup P$ receives $R(\expectB, \actualU_i)$, the set of out-neighbors $v$ of $Z$ but not in $Z$ validly relaxed from complete vertex $u\in \actualU_i$ with $d[u] + w_{uv} = d[v] < \expectB$:
    \begin{itemize}
        \item If $d[u] + w_{uv} < \expectB_i$, $v$ is included by $\ds_i$ and merged into $\ds$ (\cref{line:bmssp:merge});
        \item If $d[u] + w_{uv} \in [\expectB_i, \expectB)$, $v$ is directly inserted into $\ds$ (\cref{line:bmssp:edge-direct-insert});
    \end{itemize}
\end{enumerate}
Therefore, at the end of the $i-$th iteration, $\frontierPair*{\actualU, \ds \cup P}$ is a new frontier for $\target_{P}$.

Now we verify that: at the end of the $i-$th iteration, $d_{\expectB}[\ds\cup P] \geq \actualB (=\actualB_{i})$:
\begin{itemize}
    \item For $\ds$, after \pull\ (\cref{line:bmssp:pull}), $d_{\expectB}[\ds] \geq \expectB_{i}$. For any $v$ added into $\ds$, $d[v]\geq \actualB_{i}$. Finally, $d_{\expectB}[\ds] \geq \actualB_{i}$;
    \item For $\hat{P}$, we picked out every $v\in \hat{P}$ with $d[v] < \expectB_{i}$ (\cref{line:bmssp:pick-from-Si}), so for the remaining vertices, $d_{\expectB}[\hat{P}] \geq \actualB_{i}$.
    \item For $P \setminus \hat{P}$, their shortest paths visit some complete vertices in $\ds \cup \hat{P}$, $d_{B}[P\setminus \hat{P}]\geq d_{\expectB}[\ds \cup \hat{P}] \geq \actualB_{i}$.
\end{itemize}

Now we verify that: at the end of the $i-$th iteration, for each non-empty $P_j$, $d[p_j] \leq d_{\expectB}[\hat{P}_j]$.
\begin{enumerate}
    \item If a new $p_j$ is newly re-selected from $P_j$ (\cref{line:bmssp:tree-pivot-insert}), clearly $p_j = d_{B}[P_j] \leq d_{B}[\hat{P}_j]$.
    \item Otherwise, $p_j$ is not in $U_i$ and not removed from $P$.
    For any $v\in \hat{P}_j$, at the beginning, $d[p_j] \leq d[v]$. 
    \begin{itemize}
        \item If $d[v]$ is modified by $i-$th sub-call, then $v\in \target_i=\target{\expectB_i, S_i}$. As $\frontierPair*{\actualU_i, \ds_i}$ is a frontier of $\target_i$:
        \begin{itemize}
            \item If $v\in U_i$, $v$ is removed from $P_j$, so $v\notin \hat{P}_j$;
            \item If the shortest path of $v$ visits some complete vertex in $\ds_i$ merged into $\ds$, $v\notin \hat{P}_j$.
        \end{itemize}
        \item If $d[v]$ is modified by valid relaxation from $u\in U_i$ where $d[u] + w_{uv} = d[v] \geq B_i$ (\cref{line:bmssp:edge-direct-insert}), then $p_j$ is chosen to be the smaller one (\cref{line:bmssp:update-pj-direct-insert}): $d[p_j] \leq d[v]$;
        \item Otherwise, $d[v]$ does not change over $i-$th iteration, at the end, $d[p_j]\leq d[v]$.
    \end{itemize}
\end{enumerate}

The induction proof of three propositions are complete.
Finally, we proceed to prove \cref{lemma:bmssp-correctness-frontier}.

Suppose in total, there are $f$ iterations.
At the end of the $f-$th iteration, we have that:
$\frontierPair{\actualU, \ds \cup P}$ is a frontier for $\target_{P}=\target{\expectB, P_{ini}}$; $d_{\expectB}[\ds \cup P] \geq \actualB$.
We verify the following statements sequentially:
\begin{itemize}
    \item At the end of the $f-$th iteration, $W' = \{x\in W \setminus \actualU: d[x] < \actualB\}$ is complete.

    If $v\in W' \subseteq W$ is incomplete, because $\frontierPair{W, P_{ini}}$ is a frontier for $\target$, $v$ visits some complete vertex in $P_{ini}$. Then $v\in \target_{P}$, so the shortest path of $v$ visits some complete vertex $u \in \ds \cup P$, such that $\distV(v) \geq \distV(u) = d[u] \geq d_{\expectB}[\ds \cup P] \geq \actualB$, but this violates with $d[v] < \actualB_{i}$ by its definition.
    
    Thus $W'$ is complete.
    
    \item At the end of the $f-$th iteration,
    for any vertex $v$ in $\target$, $v$ is complete, or the shortest path of $v$ visits some complete vertex in $\ds \cup P$; or the shortest path of $v$ visits some complete vertex in $q\in Q$, such that $d[q] = \distV(q) < \actualB$ and $q\in W'$, or $d[q] = \distV(q) \in [\actualB, \expectB)$.

    In other words:
    $\frontierPair{\actualU, \ds \cup P \cup W' \cup Q^*}$ is a frontier for $\target$, where $Q^* = \{x\in Q: d[x] \in [\actualB, \expectB)\}$.

    For any $v\in \target$, if $v\in \target_{P}$, then $v$ is handled by $U$ and $\ds \cup P$.
    If $v\in \target \setminus \target_{P}$, the shortest path of $v$ visits some $q\in Q= S \setminus P_{ini} \subseteq W$ which has been complete at the beginning of this call.
    Currently $\actualU \subseteq \target_{P}$, so $q\notin \actualU$. Thus, $q\in W'$ or $q\in Q^*$.
    
    \item After \cref{line:bmssp:add-back-Si}, $\frontierPair{\actualU, \ds \cup W'}$ is a frontier of $\target$.
    
    For every vertex $v\in P$, $d[v]\geq d_{\expectB}[P]\geq \actualB$ and is inserted back to $\ds$ (\cref{line:bmssp:add-back-Si}). 
    For every vertex $v\in Q^*$, as $d[v]\geq \actualB$, $v$ is also inserted back to $\ds$ (\cref{line:bmssp:add-back-Si}). 

    \item When \cref{alg:BMSSP} returns, $d_{\expectB}[\ds] \geq \actualB$;
    
    Because all values inserted into $\ds$ after the $f-$th iteration (after \cref{line:bmssp:while-end}) are at least $\actualB$, $d_{\expectB}[\ds] \geq \actualB$.
    
    \item When \cref{alg:BMSSP} returns, $\frontierPair{\actualU, \ds}$ is a frontier of $\target$.
    
    The remaining part of the algorithm (from \cref{line:bmssp:define-w-prime} to \cref{line:bmssp:insert-w}) is indeed using complete set $W'$ to update its frontier $\frontierPair{\actualU, \ds \cup W'}$ to $\frontierPair{\actualU, \ds}$ as described in \cref{para:frontier-and-dijkstra-like-alg-framework}: $\max_{x\in W'}d[x] < \actualB \leq d_{\expectB}[\ds]$ so no need to remove $W'$ from $U$; all out-neighbors of $W'$ but not in $W'$ with distance less than $\expectB$ but validly relaxed by $W'$, i.e., vertices in $R(\expectB, W')$, are inserted into $\ds$.

    \item When \cref{alg:BMSSP} returns,  $U$ is complete;
    
    % As $d_{\expectB}[\ds] \geq \actualB$, by \cref{lemma:frontier-observation}, $\actualU$ is complete.

    By induction, $U_i$ are complete; and we have shown $W'$ is complete, so $U$ is complete.
    
    \item When \cref{alg:BMSSP} returns, $\max_{x\in \actualU}\distV(x) < \actualB$ and $\actualU = \target{\actualB, S}$.
    
    By induction,
    $\max_{x\in \actualU}\distV(x) \leq \max\{\max_{x\in \actualU_{f}}\distV(x), \max_{x\in W'}\distV(x)\} < \actualB_{f} = \actualB$.

    By \cref{lemma:frontier-observation} and since $\frontierPair{\actualU, \ds}$ is a frontier for $\target$ such that $d_{\expectB}[\ds] \geq \actualB$, we have $\target{\actualB, S} \subseteq \actualU$;
    on the other hand, since $\max_{x\in \actualU}\distV(x) < \actualB$, $\actualU \subseteq \target{\actualB, S}$. Thus $\actualU = \target{\actualB, S}$.

    \item When \cref{alg:BMSSP} returns, $\ds$ contains $S \setminus \actualU$ and every out-neighbor $v$ of $\actualU$ but not in $\actualU$ validly relaxed from $u\in \actualU$ with $d[u] + w_{uv} = d[v] < \expectB$, i.e., every vertex in $R(\expectB, \actualU)$.

    % As all values inserted into $\ds$ after the $f-$th iteration are at least $\actualB$,
    % Moreover, by definition, $\max_{x\in \actualU}\distV(x) \leq \max\{\max_{x\in \actualU_{i}}\distV(x), \max_{x\in W'}\distV(x)\} \leq \actualB$.
    
    For any $v\in S\subseteq \target$, if $d[v] < \actualB \leq d_{\expectB}[\ds]$, then $v\in \target{\actualB, S} \subseteq \actualU$.

    Thus $S \setminus \actualU \subseteq \{x\in S: d[x] \in [\actualB, \expectB)\} \subseteq \ds$.
    
    For every out-neighbor $v$ of $U$ but not in $\actualU$ validly relaxed from $u\in U$  with $d[u] + w_{uv} = d[v] < \expectB$,
    by $\actualU = \target{\actualB, S}$ we know that 
    $d[v] \geq \actualB$;
    and we can verify $v\in \ds$ by direct insertion or merge (\cref{line:bmssp:edge-direct-insert}, \cref{line:bmssp:merge}, \cref{line:bmssp:insert-w}).
    Thus, $\ds \supseteq (S \setminus \actualU) \cup R(\expectB, \actualU)$.
    
\end{itemize}
\end{proof}

\begin{lemma}[Correctness of BMSSP, focusing on size constraints]
\label{lemma:bmssp-correctness-size}
Under the same conditions of \cref{lemma:bmssp-correctness-frontier}, $\abs{\actualU} = O(t^32^{lt})$. Especially, for partial execution, $\abs{\actualU} = \Theta(t^32^{lt})$. Every insertion in \cref{alg:BMSSP} takes $O(t)$ time.

For $\actualU_{i}$'s returned by the sub-calls under one single call of \BMSSP, they are disjoint. As a result, $\actualU$ is a disjoint union of $\actualU_{i}$'s and $W'$.
\end{lemma}

\begin{proof}
$\abs{\actualU}$ constraints follow from the while-condition of \cref{alg:BMSSP} and induction: $t^32^{lt} + t^32^{(l-1)t} = \Theta(t^32^{lt})$. For insertion time, as $\abs{\actualU}\leq O(t^32^{lt})$, the number of vertices added into $\ds$ is bounded by $\abs{S} + \abs{\actualU}\delta \leq O(t^42^{lt})$ ($\delta \leq k \leq t$). With $M = t2^{(l-1)t}$, each insertion takes $O(\log (t^42^{lt}/M)) = O(t)$ time.

For $U_i$ disjointness, we only need to show: $\{\distV(v): v\in \actualU_{i}\} \subseteq [\actualB_{i-1}, \actualB_{i})$.
By \cref{lemma:bmssp-correctness-frontier}, $\max_{u\in \actualU_{i}} \distV(u) < \actualB_{i}$.
Also by \cref{lemma:bmssp-correctness-frontier}, in the $i-$th iteration,
$\distV_{\expectB}(\actualU_{i}) \geq \distV_{\expectB}(S_i) = d_{\expectB}[S_i] \geq d_{\expectB}[\ds \cup P] \geq \actualB_{i-1}$.
% Because the sub-calls satisfies \cref{lemma:bmssp-correctness-frontier} and that $\actualU_{i} = \target{\actualB_{i}, S_i}$, by definition $\max_{v\in \actualU_{i}} \distV(v) < \actualB_{i}$.
% Because the sub-calls satisfies \cref{lemma:bmssp-correctness-frontier}, $\max_{u\in \actualU_{i}} \distV(u) < \actualB_{i}$.
\end{proof}

\subsection{Time analysis of BMSSP}
Following \cite{SSSP25}, for a vertex set $\actualU$ and two bounds $c < d$,
% denote $N^+(\actualU)=\{(u, v\in E: u\in \actualU)\}$ to be the set of outgoing edges from $\actualU$,
denote $N^+_{[c, d)}(\actualU) = \{(u, v)\in E: u\in \actualU, \distV(u)+w_{uv}\in [c, d)\}$.
Clearly $\abs*{N^+_{[c, d)}(\actualU)} \leq \delta \abs{\actualU}$.

Denote $E(\actualU) = \{(u, v)\in E: u, v\in \actualU\}$.
% which is the set of edges within the induced subgraph $G|_{\actualU}$ on $\actualU$.
Clearly $\abs{E(\actualU)} \leq \delta \abs{\actualU}$.

Recall in \cref{subsec:obserevation} we defined the recursion tree $\mathcal{T}$ of \cref{alg:BMSSP}.
For $X$ a \BMSSP call, we subscript $X$ to all the parameters to denote those used in $X$.
Denote $\mathcal{T}(X) = \{X': X' \text{ is in the subtree rooted at }X\}$ and $\mathcal{T}^{\circ}(X) = \mathcal{T}(X)\setminus \{X\}$.

\begin{lemma}[Time analysis of BMSSP]
\label{lemma:bmssp-time-analysis}
Denote $C$ to be a sufficiently large time constant.
Under the same conditions with \cref{lemma:bmssp-correctness-frontier}, a call $X$ of \BMSSP runs in time: (here $\actualU=\actualU_X$)
\[C\abs{\actualU}(l+\log t)\log t + Ct\abs*{N^+_{[\distV_{\expectB}(S), \expectB)}(\actualU)} + Ct\abs{E(\actualU)} + Ct\sum_{Y \in \mathcal{T}(X)} \abs{Q_{Y}},\]
Moreover, $\sum_{Y\in \mathcal{T}(X)}\abs{Q_{Y}} \leq \abs{\actualU}(\delta + l / t)$.

Thus in the top layer when called with $\expectB=\infty, S=\{s\}$ and $l = \greaterInt{(\log n)/t}$, it takes time $O(m(t + \frac{\log n\log t}{\delta t}))$, optimized to $O(m\sqrt{\log n}+\sqrt{mn\log n\log \log n})$ with best parameters.
\end{lemma}
\begin{proof}
We first prove the time of one \BMSSP call by induction.
Clearly base case satisfies \cref{lemma:bmssp-time-analysis}. 

Now assume $l > 0$.
By induction, sub-calls satisfy \cref{lemma:bmssp-time-analysis} and take time
\[T_0 = C\sum_{i}\left(\abs{\actualU_{i}}(l - 1 + \log t) \log t + t\abs*{N^+_{[\distV_{\expectB_{i}}(S_i), \expectB_{i})}(\actualU_{i})} + t\abs{E(\actualU_{i})}\right) + Ct\sum_{Y\in \mathcal{T}^{\circ}(X)}\abs{Q_{Y}}.\]
By \cref{lemma:bmssp-correctness-size}, $\sum_{i}\abs{\actualU_{i}} \leq \abs{\actualU}$.
By \cref{lemma:bmssp-correctness-frontier},
in the $i$-th iteration,
$\frontierPair{\emptyset, S_i}$ is a frontier for $\target{\expectB_{i}, S_i}$ and $S_i$ is chosen from $\ds \cup P$,
% and by definition $d_{\expectB}[P \setminus \hat{P}] \geq d_{\expectB}[D]$,
so $\distV_{\expectB_{i}}(S_i) = d_{\expectB_{i}}[S_i] \geq d_{\expectB}[\ds \cup P]
% \geq d_{\expectB}[\ds \cup \hat{P}]
\geq \actualB_{i-1}$. Thus
\[ T_0 \leq C\abs{\actualU}(l - 1 + \log t) \log t + Ct\sum_{i}\abs*{N^+_{[\actualB_{i-1}, \expectB_{i})}(\actualU_{i})} + Ct\sum_{i}\abs{E(\actualU_{i})} +Ct\sum_{Y\in \mathcal{T}^{\circ}(X)}\abs{Q_{Y}}.\]
% For $\actualB_{0}=\min\{\expectB, \min_{1\leq j\leq p}d[p_j]\}$ as defined in \cref{line:bmssp:define-B-zero}, also $\actualB_{0} \geq \min_{x\in S}d[x] = \min_{x\in S}\distV(x)$.

In a full execution, $\actualU = \target$.
In a partial execution, $\abs{\actualU} > t\abs{S}$.
In both cases, $\abs{\actualU} \geq \min\{t\abs{S}, \abs*{\target}\}$.

We pick $C' = C/100$ to bound the time constant directly consumed by $X$ and list them sequentially:
% Now we list the time of the rest of \cref{alg:BMSSP} as follows:
\begin{enumerate}
    \item By \cref{remark:find-pivots}, $p\leq \min\{\abs{S}, \abs*{\target}/k\} \leq \abs{\actualU}/k$. Thus \FindPivots(\cref{line:bmssp:find-pivot-call}) and initial insertion (\cref{line:bmssp:initial-pivot-insert}) take time $T_1 = C'\abs{\actualU} \log t +C' t \abs{Q}$.
    \item Picking $\{x\in P_j: d[x] < \expectB_{i}\}$ (\cref{line:bmssp:pick-from-Si}) takes $C'k$ time for each $i$ and $j$ such that $p_j \in S_i$.
    \begin{itemize}
        \item If $i-$th sub-call is a partial execution, $T_{2, par, i} \leq C'k \abs{S_i} \leq C'\abs{\actualU_i}$.
        \item If $i-$th sub-call is a full execution, $p_j\in \actualU_{i}$ is removed from $P_j$ (\cref{line:bmssp:pivot-remove-add-j}), causes re-selection and insertion of new $p_j$ (\cref{line:bmssp:tree-pivot-insert}) of $O(t)$ time and can be absorbed by it (\cref{item:lemma-time:reselection}); except if $P_j$ becomes empty, but this happens at most one time for each $j$: $T_{2, full} \leq C' p k = C' \abs{\actualU}$.
    \end{itemize}
    In total $T_2 \leq C'\sum_{i}\abs{\actualU_i} + C' \abs{\actualU} = 2 C' \abs{\actualU}$.
    \item Iterating $\actualU_{i}$, updating $\dirtyPj$, \pull\ and \mergeds\ operations and checking edges of $\actualU_i$ and $W'$ (Line~\ref{line:bmssp:pivot-remove-add-j},\ref{line:bmssp:pull},\ref{line:bmssp:merge},\ref{line:bmssp:edge-direct-check}, \ref{line:bmssp:check-w}) together take linear time $T_3 = C'\delta\abs{\actualU} \leq C'\abs{\actualU}\log t$.
    \item Insertion of $v$ from edges $(u, v)$ in $N^+_{[\expectB_{i}, \expectB)}(\actualU_{i})$ (\cref{line:bmssp:edge-direct-insert}) takes time $T_4 = C't\sum_{i}\abs*{N^+_{[\expectB_{i}, \expectB)}(\actualU_{i})}$.
    \item Re-selection and insertion of each new $p_j$ (\cref{line:bmssp:tree-pivot-insert}) take time $O(k+t) \leq C't$. \label{item:lemma-time:reselection}
    \begin{itemize}
        \item If $X$ is a partial execution, $T_{5, par} = C't\abs{S} \leq C'\abs{\actualU}$.
        \item Below we assume $X$ is a full execution.
        For each $P_j$, in the $i-$th iteration, if it spends $C't$ time (\cref{line:bmssp:tree-pivot-insert}) to re-select and insert a new $p_j$, its old $p_j'$ must have been in $\actualU_{i}$ and removed (\cref{line:bmssp:pivot-remove-add-j}).
        For each $j$, $P_j$ has at most one such $p_j'$ removed from each $\actualU_{i}$.
        Denote $g_j$ the number of such $p_j'$'s, and then the time is $T_{5, full} = C't\sum_{j=1}^{p}g_j$.
        By \cref{lemma:bmssp-correctness-size}, $\actualU_{i}$'s are disjoint.
        For each $j$, those $p_j'$'s over different iterations distribute in $g_j$ distinct $\actualU_{i}$'s.
        
        By \cref{lemma:find-pivots}, each $P_j$ is connected by an undirected tree $F_j$ whose vertices are in $\target = \actualU$.
        Denote $\mathcal{E} = E(\actualU) \setminus (\bigcup_{i}E(\actualU_{i}))$, the set of edges connecting vertices from two different $\actualU_{i}$'s.
        Removing edges of $\mathcal{E}$ from $F_j$ splits it into at least $g_j$ weakly connected components. Thus $F_j$ contains at least $g_j -1$ edges in $\mathcal{E}$.
        As $F_j$'s are edge-disjoint and $tp \leq \abs{\actualU}t/k = \abs{\actualU}\log t$,
        \[T_{5, full} = C't\sum_{j=1}^{p}g_j \leq C'tp + C't\sum_{j=1}^{p}\abs{F_j \cap \mathcal{E}} \leq C'\abs{\actualU}\log t + C't\abs{\mathcal{E}}.\]
    \end{itemize}
    \item Insertion of vertices in $\{x\in S: d[x] \in [\actualB, \expectB)\}$ (\cref{line:bmssp:add-back-Si}) takes time $O(t\abs{S})$, 
    but it is non-zero only if this call is a partial execution ($\actualB < \expectB$) such that $\abs{S} \leq \abs{\actualU} / t$, so $T_6 = C't\abs{S} \leq C'\abs{\actualU}$.
    \item Insertion of $v$ from edges $(u, v)$ in $N^+_{[\actualB, \expectB)}(W')$ (\cref{line:bmssp:insert-w}) takes time $T_7 = C't\abs*{N^+_{[\actualB, \expectB)}(W')}$.
\end{enumerate}

To sum up, because $\actualB_{i-1} \geq \actualB_{0} \geq d_{\expectB}[S] = \distV_{\expectB}(S)$, $\actualU$ is a disjoint union of $\actualU_{i}$'s and $W'$, and $\abs{\mathcal{E}}+\sum_{i}\abs{E(\actualU_{i})} \leq \abs{E(\actualU)}$, the time of one \BMSSP call is
\begin{align*}
T \leq\ & T_0 + T_1 + T_{2} + T_3 + T_4 + T_{5, par} + T_{5, full} + T_6 + T_7 \\
\leq\ & \left(C\abs{\actualU}(l - 1 +\log t)\log t + 7 C'\abs{\actualU}\log t\right) \\
&+ \left(C't\abs*{N^+_{[\actualB, \expectB)}(W')} + \sum_i\left(Ct\abs*{N^+_{[\actualB_{i-1}, \expectB_{i})}(\actualU_{i})} + C't\abs*{N^+_{[\expectB_{i}, \expectB)}(\actualU_{i})} \right) \right)\\
&+ \left(C't\abs{\mathcal{E}} + Ct\sum_{i}\abs{E(\actualU_{i})}\right)
+ \left(Ct \sum_{Y\in \mathcal{T}^{\circ}(X)}\abs{Q_{Y}} + C't\abs{Q_{X}}\right)\\
\leq\ & C\abs{\actualU}(l + \log t)\log t + Ct\abs*{N^+_{[\distV_{\expectB}(S), \expectB)}(\actualU)} + Ct\abs{E(\actualU)} + Ct\sum_{Y\in \mathcal{T}(X)}\abs{Q_{Y}}.
\end{align*}

Now we show that, $\sum_{Y\in \mathcal{T}(X)} \abs{Q_{Y}} \leq \abs{\actualU}(\delta + l / t)$.
\begin{itemize}
    \item For a partial execution $Y$, $\abs{Q_{Y}} \leq \abs{\actualU_{Y}} / t$.
    The $\actualU_{Y}$'s for $Y$ under $\mathcal{T}(X)$ on each layer are disjoint and their union is a subset of $\actualU$.
    Summing over all layers we obtain $\sum_{partial\ Y\in \mathcal{T}(X)}\abs{Q_{Y}} \leq \abs{\actualU}l/t$.

    \item Below we assume $Y$ is a full execution.
    For any $v\in Q_{Y}$, think of the latest valid relaxation (\cref{line:bmssp:edge-direct-insert} or \cref{line:bmssp:insert-w} in \cref{alg:BMSSP}, or \cref{line:base:insert} in \cref{alg:base-case}) on edge $(u, v)$ for some complete vertex $u$ which caused insertion of $v$.
    Then $u$ is already complete before the execution of $Y$.
    
    If $v\in Q_{Z}$ for another full execution $Z$, since $v\in \actualU_Y$ and also $v\in \actualU_Z$, $Z\in \mathcal{T}(Y)$ or $Y\in \mathcal{T}(Z)$.
    Without loss of generality, assume $Z\in \mathcal{T}(Y)$.
    Because $v\in Q_Y$ is not passed down to descendants of $Y$ in $\mathcal{T}(Y)$, and also $u$ is not in $\actualU_{Y}$, so $v\in Q_{Z}$ must be owing to another edge $(u', v)$ relaxed in some descendant of $Y$.
    Therefore, the number of appearances of $v$ in $Q_{Y}$ of full executions $Y$ is bounded by the number of relaxations from its complete predecessors, which is bounded by $\delta$.
    
    Thus $\sum_{full\ Y\in \mathcal{T}(X)}\abs{Q_Y} \leq \delta \abs{\actualU}$.
\end{itemize}

To sum up, note that both $\abs{N^+(\actualU)}$ and $\abs{E(\actualU)}$ are bounded by $O(\delta \abs{U})$, a call $X$ of \BMSSP runs in time $O(\abs{U}(l\log t + \delta t))$.
Recall we work on a graph with $O(m)$ edges, $O(m/\delta)$ vertices and max degree $\delta$, with $l = \greaterInt{(\log n) / t}$,
so in the top layer the total time is $O(m(t + \frac{\log n\log t}{\delta t}))$.

The best\footnote{``best'' in terms of order analysis, the specific constant, though not a concern of this paper, may be adjusted in implementations.} parameter is $t = \greaterInt{\sqrt{\log n \log \log n / \delta}}$ and $\delta = \frac{1}{4}\min\{\frac{m}{n}, \log \log n\}$.
If $\frac{m}{n} \leq \log\log n$, the total time is $O(\sqrt{mn\log n\log \log n})$.
If $\frac{m}{n} > \log \log n$, the total time is $O(m\sqrt{\log n})$.
\end{proof}

\subsection{Base Case}
\label{subsec:base-case}
The base case is solved using Dijkstra's algorithm.
It is efficient enough to spend $\log t$ time on each edge.

\begin{algorithm}[H]
\caption{Base Case of BMSSP (\cref{alg:BMSSP}) when $l=0$}
\label{alg:base-case}
%\AlgInput{Same as \cref{alg:BMSSP}, except that $l = 0$ and $\abs{S} \leq t^2$.}
%\AlgOutput{Same as \cref{alg:BMSSP}.}
Initialize normal binary search tree $\ds \gets \{\pair{x, d[x]}: x\in S\}$ (parameterized by $M=1$)\;
$\actualU \gets \emptyset$\;
\While(\tcp*[f]{$\abs{\ds} \leq t^3\delta \leq t^4$, each operation $O(\log t)$ time.}){$\ds.\nonEmpty()$ \LogicAnd $\abs{\actualU} \leq t^3$ }{
    $u, \actualB \gets \ds.\pull()$\;
    $\actualU \gets \actualU \cup \{u\}$\;
    \ForAll(\tcp*[f]{$u$ is complete}){edge $(u, v)$}{
        \If{\Relax{$u, v, \expectB$}\label{line:base:relax}}{
            $\ds.\ins(v, d[v])$\label{line:base:insert}\tcp*[r]{On duplicate, keep the smaller one.}
        }
    }
}
% $\ds.\promote(t)$\;
\Return $\actualB$, $\actualU$, $\ds$
\end{algorithm}

\appendix

\section{Missing Proofs}
\subsection{Partition}
\label{sec:partition}
In this section we introduce the tree partition algorithm used for pivot selection.
We slightly modify the topological partition algorithm in \cite{Frederickson83}, which is also used in \cite{SSBP18}. Note that in the algorithm $T$ is a directed tree but we treat every edge in $T$ as undirected.

\begin{algorithm}[ht]
\caption{Tree partition}
\label{alg:partition}
\Fn{\Partition{tree $T$}}{
    $r \gets$ root of $T$\;
    $U \gets \{r\}$\;
    \ForAll{subtree $T'$}{
        $U \gets U\cup \Partition(T')$\;
        \If{$\abs{U}\geq s$}{
            Report a new group $U$\;
            $U \gets \{r\}$\;
        }
    }
    \Return $U$
}
\end{algorithm}

\begin{lemma}[Partition]
\label{lemma:partition}
    Given a tree $T$ with $n$ vertices and an integer $s\in[1,n]$, we can partition $T$ into edge-disjoint subtrees $T_1,T_2,\cdots,T_p$ in linear time, such that $|T_j|\in[s,3s)$ for all $j$.
\end{lemma}
\begin{proof}
    Algorithm \ref{alg:partition} outlines the tree partition procedure. We select an arbitrary vertex as the root.
    The main algorithm calls \Partition{$T$} that performs a depth-first search on $T$, and collects all the groups reported as $T_1,T_2,\cdots,T_p$ in order. The $\actualU$ returned by \Partition{$T$} is merged into $T_p$, or regarded as a new group if no group is reported before.
    
    \Partition{$T$} always returns a subset of vertices in $T$, including the root $r$. Prior to the merge, by induction we can deduce that each $T_j$ and the final $\actualU$ induce a subtree, and each edge in $T$ appears in exactly one of them. 
    If the merge occurs, let $r_p$ be the root node when $T_p$ is reported, then $r_p$ must be contained in both $T_p$ and $\actualU$, thus $T_p$ after the merge can still induce a subtree in $T$.
    
    It's clear that the set $\actualU$ returned in line 9 has size $[1,s]$, therefore any group reported has size $[s,2s)$. After the final merge the size of $T_p$ remains less than $3s$.
\end{proof}

\subsection{Proof of the Data Structure}
\label{subsec:proof-of-data-structure}
\begin{lemma}[\cref{lemma:data-structure} restated]
\label{lemma:data-structure-restate}
Given at most $N$ key/value pairs involved, a parameter $M$, and an upper bound $\expectB$, there exists a data structure $\ds$ that supports the following operations:
\begin{description}
    \item[Insert] Insert a key/value pair. If the key already exists, keep the one with smaller value.
    \item[Merge] For another data structure $\ds'$ with smaller parameter $M'(< M/3)$ as specified in this lemma, such that all the values in $\ds'$ are smaller than all the values in $\ds$, insert all pairs of $\ds'$ into $\ds$ (keep the smaller value for duplicate keys).
    % If $\ds'$ is empty, do nothing; otherwise, we require $\abs{\ds'} = \Omega(M)$.
    \item[Pull] Return a subset $S'$ of keys where $\abs{S'}\leq M$ associated with the smallest $\abs{S'}$ values and an upper bound $x$ that separates $S'$ from the remaining values in $\ds$. Specifically, if there are no remaining values, $x = \expectB$; otherwise, $\abs{S'} = M$.
\end{description}
% \begin{description}
%     \item[Insert] Insert a key/value pair in $O(\log(N/M))$ amortized time. If the key already exists, keep the smaller.
%     \item[Merge] For another data structure $\ds'$ with smaller parameter $M'(< M/3)$ as specified in this lemma, such that all the values in $\ds'$ are smaller than all the values in $\ds$, insert all elements of $\ds'$ into $\ds$ (keep the smaller value if there are duplicate keys) in $O(\abs{\ds'})$ amortized time.    
%     \item[Pull] Return a subset $S'$ of keys where $\abs{S'}\leq M$ associated with the smallest $\abs{S'}$ values and an upper bound $x$ that separates $S'$ from the remaining values in $\ds$, in $O(\abs{S'})$ amortized time. Specifically, if there are no remaining values, $x$ should be $\expectB$. Otherwise, we have $\abs{S'} = M$.
% \end{description}

If $M = 1$, \ins\ and \pull\ take $O(\log N)$ time with \mergeds\ unsupported, for base case only (\cref{subsec:base-case}).

If $M > 1$, we require $M\geq \log(N / M)$. \ins\ takes $O(\log (N/M))$ amortized time; other operations are linear: \mergeds\ takes $O(\abs{\ds'})$ time and \pull\ takes $O(\abs{S'})$ time.
\end{lemma}

\begin{proof}
We use a self-balanced binary search tree (e.g. Red-Black Tree \cite{red-black}) of blocks of size $\Theta(M)$.

If $M=1$ in base case, a normal BST suffices, and below we assume $M \leq \log(N/M)$.

We require blocks holding disjoint intervals that values stored in each block suit its interval, and blocks are ordered by their intervals.
Inside each block, key/value pairs are maintained by unordered linked lists.

We store a table from key to the block linked list node that corresponds to it, to support $O(1)$ membership inspection and deletion;
keep track of 
the first BST block \footnote{
This is used when $\abs{\ds'}$ is too small described in the merge section.
%and should be easy to implement.
}
for $O(1)$ access 
and the total number of key/value pairs present in $\ds$ used below.
In the beginning, $\ds$ is initialized with a single empty block of $[0, \expectB)$.

\begin{description}
    \item[Insert] Check if $v$ already exists, and continue only if the new value is smaller.
    Insert the pair to the new block and delete its old occurrence (if any). This takes $O(\log(N/M))$ time.
    \item[Merge] If $\ds'$ is empty, do nothing. If $0 < \abs{\ds'} < M$, restructure $\ds'$ to one block and append to the first block of $\ds$,
    possibly inducing block normalization introduced below. Now we assume $\abs{\ds'} \geq M$.
    
    Scan blocks of $\ds'$ in ascending order, keep the smaller value for duplicates;
    whenever we collect at least $M/3$ elements or $\ds'$ is exhausted, insert them into $\ds$.
    Causing $O(\abs{\ds'} / M)$ insertions of $O(\log (N/M))=O(M)$ time each, in total we spend $O(\abs{\ds'})$ time.
    \item[Pull] If $\ds$ contains no more than $M$ elements, pulling all of them takes linear time.
    Otherwise, continuously extract the smallest block from $\ds$ until we have at least $M+1$ but at most $2M$ elements.
    Find the $(M+1)-$th element $x$ in linear time using median-find algorithm (\cite{median-find}).
    Output the smaller $M$ elements and $x$.
    Join remaining elements larger than $x$ to the current smallest block.
    If the size of that block exceeds $M$, split it into two blocks and insert them back to $\ds$.
    All operations above take $O(\log (N/M) + M) = O(M)$ as $\log (N/M) \leq M$.
    \item[Normalization] We normalize block sizes for better performance guarantees after any operation above, resulting in $O(1)$ extra time per operated key/value pair.
    After \ins\ and \mergeds\ operation, we keep the operated block size in interval $[M/3, M]$; after \pull\ and normalization, in interval $[M/2, 2M/3]$;
    or if in total there are less than $M/3$ elements, restructure $\ds$ to contain only one block.
    
    Each inserted and merged pair results in at most one pair membership change in at most two blocks.
    When the constraint is violated, the block has interacted with $\Theta(M)$ newly inserted or merged pairs since last pull or normalization.
    % {\color{red} deletion?} \lh{I mentioned ``membership change'' which includes insertion and deletion}
    We join it to an adjacent block if it is under-sized,
    or split the block into two equally sized blocks using median-find algorithm if it is over-sized.
    In total each normalization takes $O(M)$ time and can be amortized to $O(1)$ time per inserted or merged pair.
\end{description}
\end{proof}
%\input{discussion}

%\newpage
\bibliographystyle{alpha}
\bibliography{references.bib}

@InProceedings{HT20,
author="Hagerup, Torben",
editor="Montanari, Ugo
and Rolim, Jos{\'e} D. P.
and Welzl, Emo",
title="Improved Shortest Paths on the Word \text{RAM}",
booktitle="Automata, Languages and Programming",
year="2000",
publisher="Springer Berlin Heidelberg",
address="Berlin, Heidelberg",
pages="61--72",
isbn="978-3-540-45022-1"
}

@article{TM00,
author = {Thorup, Mikkel},
title = {On \text{RAM} Priority Queues},
journal = {SIAM Journal on Computing},
volume = {30},
number = {1},
pages = {86-109},
year = {2000},
doi = {10.1137/S0097539795288246},

eprint = { 
    
        https://doi.org/10.1137/S0097539795288246
    
    

}
}

@article{FW93,
title = {Surpassing the information theoretic bound with fusion trees},
journal = {Journal of Computer and System Sciences},
volume = {47},
number = {3},
pages = {424-436},
year = {1993},
issn = {0022-0000},
doi = {https://doi.org/10.1016/0022-0000(93)90040-4},
author = {Michael L. Fredman and Dan E. Willard}
}

@Article{Dij59,
  author =   "E.~W. Dijkstra",
  title =    "A note on two problems in connexion with graphs",
  journal =  "Numerische Mathematik",
  volume =   "1",
  pages =    "269--271",
  year =     "1959",
}

@Article{FT87,
  author =       "M. L. Fredman and R. E. Tarjan",
  title =        "Fibonacci Heaps and Their Uses in Improved Network
                 Optimization Algorithms",
  pages =        "596--615",
  journal =      "JACM",
  volume =       "34",
  number =       "3",
  year =         "1987",
}

@article{PR05,
author = {Pettie, Seth and Ramachandran, Vijaya},
title = {A Shortest Path Algorithm for Real-Weighted Undirected Graphs},
journal = {SIAM Journal on Computing},
volume = {34},
number = {6},
pages = {1398-1431},
year = {2005},
doi = {10.1137/S0097539702419650},
eprint = { 
        https://doi.org/10.1137/S0097539702419650
}

}

@article{Thorup00, author = {Thorup, Mikkel}, title = {Undirected Single-Source Shortest Paths with Positive Integer Weights in Linear Time}, year = {1999}, issue_date = {May 1999}, publisher = {Association for Computing Machinery}, address = {New York, NY, USA}, volume = {46}, number = {3}, issn = {0004-5411}, doi = {10.1145/316542.316548}, journal = {J. ACM}, month = {May}, pages = {362–394}, numpages = {33} }

@article{Thorup04,
title = {Integer priority queues with decrease key in constant time and the single source shortest paths problem},
journal = {Journal of Computer and System Sciences},
volume = {69},
number = {3},
pages = {330-353},
year = {2004},
note = {Special Issue on STOC 2003},
issn = {0022-0000},
doi = {https://doi.org/10.1016/j.jcss.2004.04.003},
author = {Mikkel Thorup}
}

@article{SSBP18,
  author    = {Ran Duan and
               Kaifeng Lyu and
               Hongxun Wu and
               Yuanhang Xie},
  title     = {Single-Source Bottleneck Path Algorithm Faster than Sorting for Sparse
               Graphs},
  journal   = {CoRR},
  volume    = {abs/1808.10658},
  year      = {2018},
  archivePrefix = {arXiv},
  eprint    = {1808.10658},
  bibsource = {dblp computer science bibliography, https://dblp.org}
}

@article{FW94,
title = {Trans-dichotomous algorithms for minimum spanning trees and shortest paths},
journal = {Journal of Computer and System Sciences},
volume = {48},
number = {3},
pages = {533-551},
year = {1994},
issn = {0022-0000},
doi = {https://doi.org/10.1016/S0022-0000(05)80064-9},
author = {Michael L. Fredman and Dan E. Willard}
}

@article{Thorup96, author = {Thorup, Mikkel}, title = {Floats, Integers, and Single Source Shortest Paths}, year = {2000}, issue_date = {May 2000}, publisher = {Academic Press, Inc.}, address = {USA}, volume = {35}, number = {2}, issn = {0196-6774}, doi = {10.1006/jagm.2000.1080}, journal = {J. Algorithms}, month = {May}, pages = {189–201}, numpages = {13}, url = {https://doi.org/10.1006/jagm.2000.1080}, }

@inproceedings{Raman96, author = {Raman, Rajeev}, title = {Priority Queues: Small, Monotone and Trans-Dichotomous}, year = {1996}, isbn = {3540616802}, publisher = {Springer-Verlag}, address = {Berlin, Heidelberg}, booktitle = {Proceedings of the Fourth Annual European Symposium on Algorithms}, pages = {121–137}, numpages = {17}, series = {ESA '96} }

@article{Raman97, author = {Raman, Rajeev}, title = {Recent Results on the Single-Source Shortest Paths Problem}, year = {1997}, issue_date = {June 1997}, publisher = {Association for Computing Machinery}, address = {New York, NY, USA}, volume = {28}, number = {2}, issn = {0163-5700}, doi = {10.1145/261342.261352}, journal = {SIGACT News}, month = {June}, pages = {81–87}, numpages = {7} }

@INPROCEEDINGS{flow,
  author={Chen, Li and Kyng, Rasmus and Liu, Yang P. and Peng, Richard and Gutenberg, Maximilian Probst and Sachdeva, Sushant},
  booktitle={2022 IEEE 63rd Annual Symposium on Foundations of Computer Science (FOCS)}, 
  title={Maximum Flow and Minimum-Cost Flow in Almost-Linear Time}, 
  year={2022},
  volume={},
  number={},
  pages={612-623},
  doi={10.1109/FOCS54457.2022.00064}}

@INPROCEEDINGS{BNW22,
  author={Bernstein, Aaron and Nanongkai, Danupon and Wulff-Nilsen, Christian},
  booktitle={2022 IEEE 63rd Annual Symposium on Foundations of Computer Science (FOCS)}, 
  title={Negative-Weight Single-Source Shortest Paths in Near-linear Time}, 
  year={2022},
  volume={},
  number={},
  pages={600-611},
  doi={10.1109/FOCS54457.2022.00063}}

@inproceedings{Frederickson83, author = {Frederickson, Greg N.}, title = {Data Structures for On-Line Updating of Minimum Spanning Trees}, year = {1983}, isbn = {0897910990}, publisher = {Association for Computing Machinery}, address = {New York, NY, USA}, doi = {10.1145/800061.808754}, booktitle = {Proceedings of the Fifteenth Annual ACM Symposium on Theory of Computing}, pages = {252–257}, numpages = {6}, series = {STOC '83} }

@INPROCEEDINGS {DMSY23,
author = {R. Duan and J. Mao and X. Shu and L. Yin},
booktitle = {2023 IEEE 64th Annual Symposium on Foundations of Computer Science (FOCS)},
title = {A Randomized Algorithm for Single-Source Shortest Path on Undirected Real-Weighted Graphs},
year = {2023},
volume = {},
issn = {},
pages = {484-492},
doi = {10.1109/FOCS57990.2023.00035},
url = {https://doi.ieeecomputersociety.org/10.1109/FOCS57990.2023.00035},
publisher = {IEEE Computer Society},
address = {Los Alamitos, CA, USA},
month = {November}
}

@INPROCEEDINGS{HHRTT24,
  author={Bernhard Haeupler and Richard Hlad\'{\i}k and V\'aclav Rozho\v{n} and Robert E. Tarjan and Jakub T\v{e}tek},
  booktitle={2024 IEEE 65th Annual Symposium on Foundations of Computer Science (FOCS)}, 
  title={Universal Optimality of Dijkstra via Beyond-Worst-Case Heaps}, 
  year={2024},
  volume={},
  number={},
  pages={},
  url={https://arxiv.org/abs/2311.11793}, }

@INPROCEEDINGS {BCF23,
author = {Bringmann, Karl and Cassis, Alejandro and Fischer, Nick },
booktitle = { 2023 IEEE 64th Annual Symposium on Foundations of Computer Science (FOCS) },
title = {{ Negative-Weight Single-Source Shortest Paths in Near-Linear Time: Now Faster! }},
year = {2023},
volume = {},
ISSN = {},
pages = {515-538},
doi = {10.1109/FOCS57990.2023.00038},
url = {https://doi.ieeecomputersociety.org/10.1109/FOCS57990.2023.00038},
publisher = {IEEE Computer Society},
address = {Los Alamitos, CA, USA},
month = {November}
}

@inproceedings{Fineman24,
author = {Fineman, Jeremy T.},
title = {Single-Source Shortest Paths with Negative Real Weights in $\tilde{O}(mn^{8/9})$ Time},
year = {2024},
address = {Chicago, IL, USA},
booktitle = {Proceedings of the 56th Annual ACM Symposium on Theory of Computing},
pages = {3–14},
numpages = {12}
}

@inproceedings{HJQ25,
author = {Huang, Yufan and Jin, Peter and Quanrud, Kent},
title = {Faster single-source shortest paths with negative real weights via proper hop distance},
year = {2025},
address = {Chicago, IL, USA},
booktitle = {Proceedings of the 2025 Annual ACM-SIAM Symposium on Discrete Algorithms (SODA)},
pages = {5239-5244},
numpages = {6}
}

@article{HJQ25b,
      title={Faster negative length shortest paths by bootstrapping hop reducers}, 
      author = {Huang, Yufan and Jin, Peter and Quanrud, Kent},
      year={2025},
      journal={CoRR},
      volume={abs/2506.00428},
      eprint={2560.00428},
      archivePrefix={arXiv},
      primaryClass={cs.DS},
      url={https://arxiv.org/abs/2506.00428}, 
}

@INPROCEEDINGS{red-black,
  author={Guibas, Leo J. and Sedgewick, Robert},
  booktitle={19th Annual Symposium on Foundations of Computer Science (sfcs 1978)}, 
  title={A dichromatic framework for balanced trees}, 
  year={1978},
  volume={},
  number={},
  pages={8-21},
  doi={10.1109/SFCS.1978.3}}

@article{median-find,
title = {Time bounds for selection},
journal = {Journal of Computer and System Sciences},
volume = {7},
number = {4},
pages = {448-461},
year = {1973},
issn = {0022-0000},
doi = {https://doi.org/10.1016/S0022-0000(73)80033-9},
url = {https://www.sciencedirect.com/science/article/pii/S0022000073800339},
author = {Manuel Blum and Robert W. Floyd and Vaughan Pratt and Ronald L. Rivest and Robert E. Tarjan},
}

@article{Bellman1958,
  title={ON A ROUTING PROBLEM},
  author={Richard Bellman},
  journal={Quarterly of Applied Mathematics},
  year={1958},
  volume={16},
  pages={87-90},
  url={https://api.semanticscholar.org/CorpusID:123639971}
}

@inproceedings{SSSP25,
author = {Duan, Ran and Mao, Jiayi and Mao, Xiao and Shu, Xinkai and Yin, Longhui},
title = {Breaking the Sorting Barrier for Directed Single-Source Shortest Paths},
year = {2025},
isbn = {9798400715105},
publisher = {Association for Computing Machinery},
address = {New York, NY, USA},
url = {https://doi.org/10.1145/3717823.3718179},
doi = {10.1145/3717823.3718179},
booktitle = {Proceedings of the 57th Annual ACM Symposium on Theory of Computing},
pages = {36–44},
numpages = {9},
location = {Prague, Czechia},
series = {STOC '25}
}

@inproceedings{Yan25,
      title={Lossless Derandomization for Undirected Single-Source Shortest Paths and Approximate Distance Oracles}, 
      author={Shuyi Yan},
      booktitle = {SIAM Symposium on Simplicity in Algorithms (SOSA 2026), Vancouver, Canada, January 12-14, 2026},
      pages = {},
      publisher = {{SIAM}},
      year = {2026}
}

@article{DI04APSP,
author = {Demetrescu, Camil and Italiano, Giuseppe F.},
title = {A new approach to dynamic all pairs shortest paths},
year = {2004},
issue_date = {November 2004},
publisher = {Association for Computing Machinery},
address = {New York, NY, USA},
volume = {51},
number = {6},
issn = {0004-5411},
url = {https://doi.org/10.1145/1039488.1039492},
doi = {10.1145/1039488.1039492},
journal = {J. ACM},
month = nov,
pages = {968–992},
numpages = {25}
}

%\appendix

\end{document}